\documentclass[sn-apa]{sn-jnl}

\usepackage{graphicx}%
\usepackage{multirow}%
\usepackage{amsmath,amssymb,amsfonts}%
\usepackage{amsthm}%
\usepackage{mathrsfs}%
\usepackage[title]{appendix}%
\usepackage{xcolor}%
\usepackage{textcomp}%
\usepackage{manyfoot}%
\usepackage{booktabs}%
\usepackage{listings}%
\usepackage{float}
\restylefloat{table}
\usepackage{graphicx}

\DeclareMathOperator*{\sm}{<} 
\DeclareMathOperator*{\bi}{>} 
\usepackage{cleveref}	
\DeclareMathOperator*{\argmax}{argmax} 
\usepackage[ruled,vlined,linesnumbered]{algorithm2e}

\newtheorem{theorem}{Theorem}
\newtheorem{proposition}[theorem]{Proposition}%
\newtheorem{lemma}[theorem]{Lemma}%

\theoremstyle{definition}{
\newtheorem{example}{Example}%
\newtheorem{remark}{Remark}%

\newtheorem{definition}{Definition}}

\usepackage{lipsum}

\newcommand\blfootnote[1]{%
  \begingroup
  \renewcommand\thefootnote{}\footnote{#1}%
  \addtocounter{footnote}{-1}%
  \endgroup
}

\begin{document}

\title[Stable Matching Games]{Stable Matching Games}








\author[1]{\fnm{Felipe} \sur{Garrido-Lucero}\blfootnote{F. Garrido-Lucero ORCID: 0000-0001-8973-8877}}\email{felipe.garrido-lucero@ut-capitole.fr}

\author[2]{\fnm{Rida} \sur{Laraki}\blfootnote{R. Laraki ORCID: 0000-0002-4898-2424}}\email{rida.laraki@um6p.ma}

\affil[1]{\orgname{IRIT, Université Toulouse Capitole}, \orgaddress{\city{Toulouse}, \country{France}}}


\affil[2]{\orgdiv{Moroccan Center for Game Theory, UM6P, }\orgaddress{\city{Rabat}, \country{Morocco}}}


\abstract{


Gale and Shapley introduced a matching problem between two sets of agents where each agent on one side has an exogenous preference ordering over the agents on the other side. They defined a matching as stable if no unmatched pair can both improve their utility by forming a new pair. They proved, algorithmically, the existence of a stable matching. Shapley and Shubik, Demange and Gale, and many others extended the model by allowing monetary transfers. We offer a further extension by assuming that matched couples obtain their payoff endogenously as the outcome of a strategic game they have to play in a usual non-cooperative sense (without commitment) or in a semi-cooperative way (with commitment, as the outcome of a bilateral binding contract in which each player is responsible for her part of the contract). Depending on whether the players can commit or not, we define in each case a solution concept that combines Gale-Shapley pairwise stability with a (generalized) Nash equilibrium stability. In each case we give necessary and sufficient conditions for the set of solutions to be non-empty and provide an algorithm to compute a solution. 

}




\maketitle

\newpage

\section{Introduction}

\subsection{Literature review}

The \cite{gale1962college} one-to-one two-sided market matching problem, known as the \textit{marriage problem}, consists in finding a \textit{stable} pairing between two finite sets $D$ and $H$ of same size, given that each agent on each side has an \textit{strict exogenous (total) preference ordering} over the agents on the other side. Formally, each agent $k \in D \cup H$ is endowed with $\bi_k$ such that, for any two agents $\ell$ and $\ell'$ on the opposite side, $\ell \bi_k \ell'$ represents that $k$ prefers to be matched with $\ell$ rather than with $\ell'$. 

The marriage problem focuses on computing a coupling $\mu$ that associates each agent on one side to at most one agent on the other side. The coupling $\mu$ is \textit{stable} if no uncoupled pair of agents both prefer to be paired together rather than with their partners in $\mu$. Formally, a coupling $\mu$ is \textit{blocked} if there exists a pair of agents $(d,h) \in D\times H$, not matched between them, such that,
\begin{align*}
    h >_d  \mu(d) \text{ and } d >_h \mu(h),
\end{align*}
where $\mu(k)$ represents the partner of $k$ in the matching $\mu$. The matching is stable if no pair blocks it. Gale and Shapley used a \textit{deferred-acceptance} algorithm to prove the existence of a stable matching for every instance. Their algorithm takes one of the sides of the market, called the proposer-side, and asks its agents to propose to their most preferred option that has not rejected them yet. Agents receiving more than one proposal accept the best one and reject all the others. The algorithm continues until all agents on the proposer side have been accepted by somebody. 




\cite{roth1992two} studied the structure of the set of stable matchings and its core. \cite{vate1989linear} characterized the set of stable matchings in the one-to-one problem as extreme points of a polytope when $|D| = |H|$. \cite{rothblum1992characterization} extended this result to the case $|D| \neq |H|$. \cite{balinski1997stable} and (\citeyear{balinski1998graphs}) proposed an elegant directed graph approach to the problem and, like Rothblum, they characterized the stable matching polytope in the one-to-one problem through linear inequalities, proving that any feasible point of the polytope is a stable matching and vice-versa. 

One of the first extensions of the marriage problem to the \textit{endogenous preferences} setting is the \textit{assignment game} of \cite{shapley1971assignment} in which agents within the same couple can make monetary transfers. The leading example is a \textit{housing market} where buyers and sellers have quasi-linear utilities. 
Allocations in the Shapley-Shubik model are stable if there is no unmatched pair and no transaction price from buyer to seller such that both agents end up strictly better off by trading. Exploiting the linearity of the payoff functions on the monetary transfers, Shapley and Shubik found stable solutions for their problem using linear programming where a pair primal-dual gives, respectively, the matching and the vector of prices. 

The assignment game belongs to the class of \textit{cooperative games with transferable utility} as agents within the same couple have to split their \textit{worth} (sum of their valuations) in such a way that nobody prefers to change their partner. Moreover, Shapley and Shubik proved that the set of stable allocations for their assignment game is exactly the \textit{Core} of the housing market problem seen as a transferable utility cooperative game. 

\cite{rochford1984symmetrically} extended the assignment game with transferable utility by allowing both agents within a couple to negotiate the division of their joint value. In contrast to \cite{shapley1971assignment}, where only the buyers have bargaining power and the optimal solution corresponds to the competitive equilibrium that maximizes the utility of the buyers, Rochford introduced the concept of \textit{symmetrically pairwise-bargained} (SPB) allocations. More specifically, agents within a couple equally divide the remaining surplus after receiving their threat-level utility. Rochford proved that an SPB allocation always exists and proposed a re-bargaining process that converges to an SPB allocation when starting from a core allocation that is optimal for one of the sides.

\cite{demange1985strategy} considered a model à la Shapley-Shubik with non-quasi-linear monetary transfers, allowed monetary transfers on both sides (from buyer to seller and vice-versa), and proved that the set of stable allocations has a lattice structure (non-emptiness of this set has been proved in \cite{crawford1981job} and \cite{quinzii1984core}). \cite{roth1988interior} extended the notion of SPB to the general model of Demange and Gale with agents having equally bargaining power. Using Tarski's fixed-point theorem, they showed that the set of SPB allocations forms a lattice as well and that starting from the most preferred core allocation for one side, Rochford's re-bargaining process converges to their most preferred SPB allocation. \cite{bennett1988consistent} extended the study by considering players with varying degrees of bargaining power and proving that any marriage market has at least one bargaining equilibrium. \cite{chiappori2016matching} consider a model with non-transferable utility where couples must agree on a risk-sharing rule. They proved that a stable matching exists, that it is generically unique, and is negatively assortative, i.e., the more risk averse agents in one side are matched with the less risk averse agents in the other side.


\subsection{Contributions}

To motivate our contribution, think to a worker hired by a firm. In addition to a good salary, the firm can offer to the employee perks such as medical insurance, gym, extra time-off, flexible schedule, childcare assistance, and/or days of remote work. Similarly, the worker can promise to work hard, learn new technologies, be flexible and respectful of the company code of conduct. Alternatively, consider the example of a university hiring a professor. The university can promise to reduce the teaching duties, provide administrative assistance, guarantee some research budget, and respect the promotion plan. On the other hand, the professor can promise to publish in top ranked journals, be a good teacher, supervise students, and apply for grants.

While couples can commit explicitly or implicitly on the specific actions that each of them has to do, we should keep in mind that each agent can individually violate the agreement (a worker may refuse to work certain days, a professor may teach badly). A rational agent will deviate from a non-binding contract whenever this is profitable, and will renegotiate a binding contract if they can improve it by unilaterally changing their commitments while remaining acceptable by the partner. There is an inherent strategic factor in these examples that cannot be analysed using the above fully cooperative models which neglect how utilities are formed and the consequences of a player unilateral deviation from the actions they control.


To study those non-cooperative aspects and to prevent such deviations, we propose the novel framework of \textit{matching games} which supposes that members of a couple obtain their utilities endogenously as the output of a strategic game they must simultaneously play at the time of the agreement. 


More precisely, we extend the literature on matching models by supposing that individual members of a couple $(d,h) \in D \times H$, obtain their payoffs as the output of a strategic game $G_{d,h} = (X_d , Y_h , f_{d,h},$ $g_{d,h})$, that they have to play, where $X_d$ is $d$'s action/strategy set, $Y_h$ is $h$'s action/strategy set, and $f_{d,h}, g_{d,h} : X_d \times Y_h \rightarrow \mathbb{R}$ are the utility functions of $d$ and $h$, respectively. Hence, if $d$ and $h$ are matched, $d$ chooses to play $x_d$ and $h$ chooses to play $y_h$, $d$'s and $h$'s final utilities are $f_{d,h}(x_d,y_h)$ and $g_{d,h}(x_d,y_h)$, respectively. An allocation of the matching game is a triple $\pi = (\mu, \Vec{x},\Vec{y})$ with $\mu$ a matching between $D$ and $H$, $\vec{x} = (x_d)_{d\in D} \in \prod_{d \in D} X_d$ a strategy profile for all agents in $D$, and $\vec{y}=(y_h)_{h \in H}\in \prod_{h \in H} Y_h$ a strategy profile for all agents in $H$. For example, a matching problem with linear transfers can be represented by a family of constant-sum games with strategies sets $X_d=Y_h=\mathbb{R}_+$, and payoff functions $f_{d,h}(x_d,y_h)=-x_d+y_h+a_{d,h}$ and $g_{d,h}(x_d,y_h)=x_d-y_h+b_{d,h}$, with $a_{d,h}$ and $b_{d,h}$ representing the monetary utility of being with the partner when there is no transfer.\footnote{The game of transfers is a constant sum game as the sum of payoffs $f_{d,h}(x_d,y_h)+g_{d,h}(x_d,y_h) = a_{d,h}+b_{d,h}$ is independent on $x_d$ and $y_h$. Strategically, this is equivalent to a zero-sum game and it is a particular instance of a strictly competitive game.}



From this model, we study two possible cases. The first one considers that matched couples \textbf{cannot commit} due to, e.g., the lack of binding contracts over the actions they intend to play. In this case, for the players not to deviate from the intended actions, these last must constitute a Nash equilibrium of their two-player game. Thus, an allocation $\pi = (\mu,\Vec{x},\Vec{y})$ will be called \textit{pairwise-Nash stable} if (a) all matched couples play a Nash equilibrium of their game (Definition \ref{def:internally_Nash_stable_matching_profile}) and (b) no pair of agents $(d,h)$ can jointly deviate to some Nash strategy profile $(x'_{d},y'_{h})$ in their game $G_{d,h}$ Pareto improving their payoffs (Definition \ref{def:externally_Nash_stable_matching_profile}). This last condition is an analog of Gale-Shapley's pairwise stability, with respect to Nash equilibria. Using a deferred-acceptance with competitions algorithm, it is proved that whenever all games $G_{d,h}$ admit a non-empty compact set of Nash equilibria, a pairwise-Nash stable allocation exists. 


The second studied case corresponds to the one in which \textbf{players can commit} e.g. by signing binding contracts. An allocation $(\mu,\Vec{x},\Vec{y})$ is called \textit{pairwise stable} (Definition \ref{def:externally_stable_matching_profile}) if no pair of agents $(d,h) \in D\times H$ can jointly deviate to some strategy profile $(x'_{d},y'_{h})$ in their game $G_{d,h}$ that Pareto improves their payoffs. The same deferred-acceptance with competitions algorithm allows us to prove that, whenever all strategic games $G_{d,h}$ have compact Pareto-optimal strategy sets and continuous payoff functions, the matching game admits a pairwise stable allocation. As each player is responsible of its own decisions and is rational, their commitments  must be chosen optimally, whenever the partner still agree to match and contract with them. 
A pairwise stable allocation $(\mu,\Vec{x},\Vec{y})$ is \textit{renegotiation proof} (Definition \ref{def:internally_stable_matching_profile}) if 
for each matched couple $(d,h)$, fixing $y_h$, $x_d$ maximizes $d$'s payoff under $h$'s participation constraint, and vice-versa. 

Solution concepts mixing cooperative and non-cooperative aspects are common on the literature of network formation games: fixing the network, players' actions must maximize their payoffs, and for each link in the network, both players must agree to form that link (see \cite{jackson1996strategic} or \cite{bich2017existence}). The concept of renegotiation-proofness has received particular interest in infinitely repeated games and mechanism design \citep{abreu1991perspective,abreu1993renegotiation,asheim2009renegotiation,farrell1989renegotiation,pearce1987renegotiation,rubinstein1992renegotiation,van1989renegotiation}. Our notion is more closely related to \cite{dewatripont1988commitment} who considers a setting in which agents can achieve agreement on contracts due to the existence of third parties.

We define a class of strategic games, called feasible games (Definition \ref{def:feasible_game}), which admit constrained Nash equilibria for all possible levels of participation constraints (Definition \ref{def:constrained_Nash_eq}) and prove that: (a) when all games $G_{d,h}$ are feasible, a new algorithm, if it converges, reaches a pairwise stable and renegotiation-proof allocation and (b) this new algorithm converges when all games are constant-sum, strictly competitive, potential or infinitely repeated. As strictly competitive games are feasible, Demange-Gale's results are recovered and refined (e.g. our algorithm not only converge to a pairwise stable allocation \`a la Demange-Gale, but to a renegotiation-proof one).

The article is structured as it follows. Section \ref{sec:model}  introduces the matching game model, the concept of pairwise stability with respect to a set of constrains, and proves existence of pairwise stable allocations thanks to a deterence-acceptance with competition algorithm.  \Cref{chapter:one_to_one_matching_games} leverages the results in \Cref{sec:model} to study matching games without commitment, while \Cref{chapter:external_stability_one_to_one_commitment} does the same for matching games with commitment. \Cref{chapter:internal_stability_one_to_one_commitment} introduces renegotiation proofness, a novel refinement, and studies its characterization, as well as the existence and algorithmic computation of renegotiation proof pairwise stable allocations. The last section concludes, gives future research lines, and discusses some open problems.

\section{Matching games}\label{sec:model}

We consider two finite sets of agents $D$ and $H$ that we refer to as doctors and hospitals. The cardinalities of $D$ and $H$ are denoted $|D|$ and $|H|$\footnote{We suppose a general model where $D$ and $H$ can have different size.}, respectively, and typical elements are denoted $d\in D$ and $h\in H$.
\medskip

\begin{definition}\label{def:matching}
\textup{A \textbf{matching} $\mu$ is a mapping between $D$ and $H$ where each agent on one side is matched to at most one agent on the other side. If $d\in H$ and $h \in H$ are matched in $\mu$, we will denote indistinctly $h = \mu(d)$ or $d = \mu(h)$.}
\end{definition}  
\medskip

When a couple $(d,h) \in D \times H$ forms, they get their payoffs as the output of a strategic game $$G_{d,h} := (X_d,Y_h, f_{d,h},g_{d,h}),$$
where $X_d, Y_h$ are the strategy sets of doctor $d$ and hospital $h$, respectively, and $f_{d,h}, g_{d,h} : X_d \times Y_h \to \mathbb{R}$ are their payoff functions. Remark agents' strategies do not depend on other agents. In particular, and for simplicity, we assume that agents can play the same strategies in all of their games. Denote by $$X:= \prod_{d \in D} X_d \text{ and } Y:= \prod_{h \in H} Y_h,$$ the spaces of strategy profiles. Further assumptions (such as compactness and continuity) over the strategy sets and payoff functions will be specified later. 
\medskip

\begin{definition}\label{def:men_and_women_matching_profiles}
\textup{A \textbf{doctors action profile} (resp. \textbf{hospitals action profile}) is a vector $\vec{x} = (x_1,...,x_{|D|}) \in X$ (resp. $\vec{y} = (y_1,...,y_{|H|})\in Y$). An \textbf{allocation} is a triple $\pi=(\mu,\vec{x},\vec{y})$ in which $\mu$ is a matching, $\vec{x}$ is a doctors action profile and $\vec{y}$ is a hospitals action profile.}
\end{definition}  
\medskip

It is natural to suppose that each agent has a utility of being single and that this utility is her/its \textbf{individually rational payoff} (IRP): she/it never accepts be matched with a partner  if the payoff of their game is less that her/its IRP. Formally, each doctor $d \in D$ (resp. hospital $h \in H$) will be attributed a value $\underline{f}_d \in \mathbb{R}$ (resp. $\underline{g}_h \in \mathbb{R}$), which constitutes the utility of being single. Given an allocation $\pi=(\mu,\vec{x},\vec{y})$, the \textbf{players utilities} are defined by,
\begin{align*}
\forall d \in D,\ f_d(\pi)&:= \left\{\begin{array}{cc}
    f_{d,\mu(d)}(x_d,y_{\mu(d)}) & \text{if $d$ is matched},\\
    \underline{f}_d & \text{otherwise}.
\end{array} \right.\\ 
\forall h \in H,\ g_h(\pi)&:= \left\{\begin{array}{cc}
     g_{\mu(h),h}(x_{\mu(h)},y_h)  &  \text{if $h$ is matched},\\
    \underline{g}_h & \text{otherwise}.
\end{array}\right.
\end{align*}

We extend the agent sets $D$ and $H$ by adding to each of them the so-called \textbf{empty players} $d_0,h_0$ who, in our future algorithms, will respect the following rules: (1) empty players have empty strategy sets and null payoff functions, (2) they can be matched with as many agents as needed but never between them, and (3) any player matched with an empty player receives her/its IRP as payoff. We denote $D_0 := D \cup \{d_0\}$ and $H_0 := H \cup \{h_0\}$.
\medskip

\begin{definition}
\textup{A tuple $\Gamma = (D_0, H_0, \{G_{d,h}: (d,h) \in D \times H\}, \underline{f}, \underline{g})$ will be called a \textbf{matching game}.}
\end{definition}
\medskip

To illustrate our model, we consider the following leading examples.
\medskip

\begin{example}\label{ex:prisoners_dilemma_example}
\textup{Consider a matching game with only one agent $d$ and one agent $h$, both having strictly positive IRPs $\underline{f}_d = \underline{g}_h = \delta > 0 $. Suppose that, if they agree to match, they play the following prisoners' dilemma $G$,
\begin{table}[H]
\centering
\begin{tabular}{cccc}
     & \multicolumn{3}{c}{Agent h}\\ \noalign{\vskip 2mm}
     \multirow{4}{*}{Agent d  } & & Cooperate & Betray \\
     \cmidrule{2-4}
     & Cooperate & $2\delta, 2\delta$ & $-\delta,3\delta$ \\ 
     \cmidrule{2-4}
     & Betray & $3\delta,-\delta$ & $0,0$
\end{tabular}
\hspace{1.5cm}
\end{table}
Notice that matching and playing the Nash equilibrium of $G$ is Pareto-dominated by remaining single. As will be seen, to have a stable allocation in which agents match, players should commit to cooperate with a sufficiently high probability.} \qed
\end{example}
\medskip


\begin{example}\label{ex:one_couple_transfer_game_example}
\textup{Consider a matching game with only one agent $d$ and one agent $h$, both having strictly positive IRPs $\underline{f}_d = \underline{g}_h = \delta > 0$. Suppose that, if they agree to match, they play a constant-sum game 
\begin{align*}
&G = (\mathbb{R}_+,\mathbb{R}_+,f_{d,h},g_{d,h}), \text{ such that for any } x,y \geq 0,\\
&f_{d,h}(x,y) = 10\delta - x + y,\\
&g_{d,h}(x,y) = x - y.
\end{align*}
Game $G$ corresponds to a transfer game in which each player increases her payoff thanks to the transfer of the partner, and decreases it due to her transfer. Since positive transfers are always strictly dominated by the null transfer, the only Nash equilibrium of $G$ is $(x,y) = (0,0)$. In that case, $h$ payoff is 0 and she is better off being single than accepting the matching and receiving a zero payoff. On the other hand, being single is Pareto-dominated by being matched and agent $d$ offering any monetary transfer $x \in [\delta, 9\delta]$ to agent $h$. For that proposal to be acceptable by $h$, the monetary transfer promised must be enforceable/credible (the commitment assumption).} \qed
\end{example}
\medskip

\begin{example}\label{ex:auction}
\textup{Consider a market with a set of $n$ buyers $N = \{1,...,n\}$, and one seller $h$ of an indivisible good. Buyers in $N$ have strictly positive values for the good, $(v_d)_{d \in N} \subseteq \mathbb{R}_+$. The seller has a reservation price $c$, strictly positive as well, and all agents have null IRPs. Since the good is indivisible, only one buyer can be matched with the seller. If a couple buyer-seller $(d,h)$ is created and the monetary transfers $(x_d, y_h)$\footnote{$x_d\geq 0$ means buyer pays to the seller, $y_h\geq 0$ means the seller pays to the buyer.} is agreed and enforceable, the item is sold from $h$ to $d$ at the price $p = x_d - y_h$. Under this scenario, the players' utilities are,
\begin{align*}
  &f_d(x_d,y_h) = v_d - p = v_d - x_d + y_h,\\
  &g_h(x_d,y_h) = p - c = x_d - y_h - c,\\
  &f_{d'} = 0 \text{ for any other buyer $d' \in N\setminus\{d\}$}.
\end{align*}
For simplicity, we assume $v_1 \geq v_2 \geq ... \geq v_n$. Moreover, we take $c \leq v_1$ as, otherwise, notice it is not rational for the agents to engage in a transaction.} \qed
\end{example}
\medskip

The first condition we will ask to any allocation is \textit{individual rationality}.
\medskip

\begin{definition}\label{eq:individual_nash_rationality}
\textup{An allocation $\pi = (\mu,\vec{x},\vec{y})$ is \textbf{individually rational} if for any $d \in D$ and any $h \in H$, $f_d(\pi) \geq \underline{f}_d$ and $g_h(\pi) \geq \underline{g}_h$.}
\end{definition}
\medskip

Pairwise stability will depend on the strategies the players are allowed to play. We will use the following general definition.
\medskip

\begin{definition}\label{def:external_stability_with_respect_to_family}
An individually rational allocation will be \textbf{pairwise stable with respect to the family} $\mathcal{C} := (C_{d,h}: d \in D, h \in H)$, where $C_{d,h} \subseteq X_d \times Y_h$ if no pair $(d,h)$ can match together, play a strategy profile in their set $C_{d,h}$, and strictly increase their payoffs.
\end{definition}
\medskip

Under mild assumptions we can show the existence of pairwise stable allocations.
\medskip

\begin{theorem}\label{teo:existence_ext_stable_alloc_general_case}
Suppose that all the sets in the family $\mathcal{C}$ are non-empty and compact, and that agents' payoff functions are continuous. Then, there exists a pairwise stable allocation with respect to the family $\mathcal{C}$.
\end{theorem}
\medskip

\Cref{teo:existence_ext_stable_alloc_general_case} is proved in two steps:
\begin{itemize}
    \item[1.] We design a deferred-acceptance with competitions (DAC) algorithm to compute an $\varepsilon$-pairwise stable allocation (\Cref{def:epsilon_externally_stable_matching_profile}). 
    \item[2.] As the sets $C_{d,h}$ are compact and the payoff functions are continuous, accumulation points as $\varepsilon \to 0$ exist and any of them will be a pairwise stable allocation.
\end{itemize}

The pseudo-code of the algorithm used in step 1 (Algorithm \ref{Algo:Propose_dispose_algo_simple_case}) is similar to the first of the two algorithms proposed by \cite{demange1986multi}. Our deferred-acceptance with competitions algorithm takes one of the sides (the doctors for the rest of the article) and asks its unmatched agents to propose, one by one, a \textit{contract} from the family $\mathcal{C}$ to their most preferred option. The proposal is computed such that the proposed agent is always better off by accepting it. If the proposed hospital is already matched, a competition between the two doctors is triggered. The winner remains and the loser goes back to the set of unmatched agents. Therefore, each iteration of Algorithm \ref{Algo:Propose_dispose_algo_simple_case} has two phases: a \textbf{proposal} and a \textbf{competition}. 

\begin{algorithm}[ht]
\textbf{Input}: $\Gamma = (D_0,H_0, (G_{d,h} : (d,h) \in D \times H), \underline{f}, \underline{g})$ a matching game, $\varepsilon \bi 0$

Set $D' \leftarrow D$ as the set of single doctors, and $g_h(\pi) \leftarrow \underline{g}_h, \forall h \in H$

\While{$D' \neq \emptyset$}{
Let $d \in D'$. Compute his \textbf{optimal proposal}  
$$(h,x,y) \in \argmax\{f_{d,h}(x,y) : g_{d,h}(x,y) \geq g_h(\pi) + \varepsilon, h \in H_0, (x,y) \in C_{d,h} \}$$ 

\If{$h$ is single}{$d$ is automatically accepted}
\Else{$d$ and $\mu(h)$ \textbf{compete} for $h$ as in a second-price auction. The winner passes to be the new partner of $h$ and the loser is included in $D'$}}
\caption{Deferred-acceptance with competitions algorithm}
\label{Algo:Propose_dispose_algo_simple_case}
\end{algorithm}

Let us explain the two phases that compose an iteration of the DAC algorithm.
\smallskip

\noindent\textbf{Proposal phase.} Let $d \in D'$ be the proposer. Given the current allocation $\pi$ (initially empty) that generates a hospitals' payoff vector $g(\pi) = (g_h(\pi))_{h \in H} = (\underline{g}_h)_{h\in H}$, $d$ computes her optimal proposal as,
\begin{align}\label{eq:problem_P_i}
    (h,x,y) \in \argmax\left\{f_{d,h}(x,y) : g_{d,h}(x,y) \geq g_h(\pi) + \varepsilon, h \in H_0, (x,y) \in C_{d,h}\right\}.
\end{align}
The solution of Problem (\ref{eq:problem_P_i}) consists in $h$, the most preferred hospital of doctor $d$, and $(x,y) \in C_{d,h}$, the strategy profile that $d$ proposes to $h$ to play. As an abuse of notation, we may call optimal proposal only to $(x,y)$, omitting the proposed hospital $h$. Problem (\ref{eq:problem_P_i}) is always feasible as $d$ can always propose to $h_0$. If $h$ is single, $d$ is automatically accepted and the algorithm picks a new proposer in $D'$.
\medskip

\noindent\textbf{Competition phase.} If the proposed agent $h$ is matched, namely with an agent $d'$, a competition between $d$ and $d'$ starts. In the stable marriage problem, the competition is the simple comparison between the places that $d$ and $d'$ occupy in $h$'s ranking. In our case, as agents have strategies, a competition is analogous to a \textit{second-price auction}. Let $\beta_d$ be the reservation price of $d$, solution to the following problem,
\begin{align}\label{eq:problem_beta_i}
    \beta_d := \max\left\{f_{d,h'}(x,y) : g_{d,h'}(x,y) \geq g_{h'}(\pi) + \varepsilon, h' \in H_0\setminus\{h\}, (x,y) \in C_{d,h'}\right\}.
\end{align}
Analogously, we compute $\beta_{d'}$. Reservation prices are the highest payoff that $d$ and $d'$ can get by matching with somebody else. In other words, these values represent the lowest payoffs that each agent is willing to accept to be with $h$. $d$'s bid $\lambda_d$ (and analogously for $d'$) is computed by,
\begin{align}\label{eq:problem_P_max}
    \lambda_d := \max\left\{ g_{d,h}(x,y) : f_{d,h}(x,y)  \geq \beta_d, (x,y) \in C_{d,h} \right\}.
\end{align}
The winner is the doctor with the highest bid. Finally, the winner, namely $d$, pays the second highest bid. Formally, $d$ solves,
\begin{align}\label{eq:problem_P_new}
    \max\left\{ f_{d,h}(x,y): g_{d,h}(x,y) \geq \lambda_{d'}, (x,y) \in C_{d,h}\right\}.
\end{align}
The loser is included in $D'$ and a new proposer is chosen.
\medskip

\begin{remark}
\textup{A defeated doctor cannot propose right away to the same hospital as she is unable to increase the hospital's payoff by $\varepsilon$. This is crucial for the convergence of the DAC algorithm.}
\end{remark}
\medskip

\begin{remark}
\textup{The output of the DAC algorithm (Algorithm \ref{Algo:Propose_dispose_algo_simple_case}) corresponds to an $\varepsilon$-approximation. This is in line with the matching literature with transfer \citep{demange1986multi,hatfield2005matching,kelso1982job}. The problem of computing a $0$-stable allocation remains open, in our case as well as in the literature with transfers.}
\end{remark}  
\medskip

We focus in proving that our DAC algorithm (Algorithm \ref{Algo:Propose_dispose_algo_simple_case}) ends in finite time and its output corresponds to an $\varepsilon$-pairwise stable allocation (\Cref{def:epsilon_externally_stable_matching_profile}).
\medskip

\begin{theorem}\label{teo:propose_dispose_algo_ends_in_finite_time}
The DAC algorithm ends in finite time.
\end{theorem}

\begin{proof}
Since the strategy sets are compact and the payoff functions are continuous, they are bounded. By construction, hospitals' payoffs strictly increase by $\varepsilon\bi 0$ with every proposal. Thus, the algorithm ends in a finite number of iterations.
\end{proof}  

\begin{remark}
\textup{Due to the monotonicity of hospitals' payoffs, once a doctor matches with $h_0$, she leaves the market and remains single forever.}
\end{remark}
\medskip

Forcing the doctors to increase hospitals' payoffs by (at least) $\varepsilon$ with every proposal guarantees the finiteness of the algorithm. However, we lose accuracy as the algorithm outputs an $\varepsilon$-pairwise stable allocation.
\medskip

\begin{definition}\label{def:epsilon_externally_stable_matching_profile}
\textup{Let $\pi = (\mu,\vec{x},\vec{y})$ be an allocation and $\varepsilon\bi 0$ fixed. A pair $(d,h) \in D \times H$ is an $\varepsilon$\textbf{-blocking pair} if there exits a strategy profile $(x'_d,y'_h) \in C_{d,h}$ such that 
$$f_{d,h}(x'_d,y'_h) > f_d(\pi) + \varepsilon \text{ and } g_{d,h}(x'_d,y'_h) > g_h(\pi)  + \varepsilon.$$
An allocation is $\varepsilon$\textbf{-pairwise stable} with respect to the family $\mathcal{C}$ if it is $\varepsilon$-individually rational (no agent gets $\varepsilon$ less than her/its IRP) and does not have any $\varepsilon$-blocking pair.}
\end{definition}

To prove the correctness of the DAC algorithm (Algorithm \ref{Algo:Propose_dispose_algo_simple_case}) we need two technical results.
\medskip

\begin{lemma}\label{prop:man_proposer_guarantees_improvement_of_epsilon}
Let $d$ be a doctor proposing to a hospital $h$, currently having a payoff $g_h$. Let $\lambda_d$ be $d$'s bid. Then, it always holds that $\lambda_d \geq g_h + \varepsilon$.
\end{lemma}  

\begin{proof}
Since $d$ proposed to $h$, there exists a strategy profile $(x,y)$ such that the triple $(h,x,y)$ is solution to Problem (\ref{eq:problem_P_i}). Therefore, the triple $(x,y, g_h + \varepsilon)$ is a feasible solution of Problem (\ref{eq:problem_P_max}). Thus, $\lambda_d \geq g_h + \varepsilon$.
\end{proof}  

\begin{lemma}\label{prop:winner_man_plays_feasible_payoff}
Let $(x,y)$ be the solution of Problem (\ref{eq:problem_P_new}). Then, $g_{d,h}(x,y)$ is always upper bounded by $\lambda_d$, $d$'s bid during the competition.
\end{lemma}  

\begin{proof}
Let $d$ (proposer) and $d'$ (current partner) be two doctors competing for $h$ and suppose, without loss of generality, that $d$ wins. Let $(x^*,y^*)$ be $d$'s optimal proposal and $(\lambda_d, \hat{x}, \hat{y})$ be the solution of Problem (\ref{eq:problem_P_max}) for player $d$. Then, the pair $(\hat{x}, \hat{y})$ is a feasible solution of $d$'s Problem  (\ref{eq:problem_P_new}), as $g_{d,h}(\hat{x}, \hat{y}) = \lambda_d \bi \lambda_{d'}$. Consider any strategy profile $(x',y')$ such that $g_{d,h}(x',y') \bi g_{d,h}(\hat{x}, \hat{y})$. If $(x',y')$ satisfies $f_{d,h}(x',y') \geq \beta_{d}$, we obtain a contradiction as $(\lambda_d, \hat{x}, \hat{y})$ is solution of Problem (\ref{eq:problem_P_max}) for player $d$. Therefore, the solution $(x,y)$ of Problem (\ref{eq:problem_P_new}) satisfies $g_{d,h}(x,y) \leq g_{d,h}(\hat{x}, \hat{y}) = \lambda_d$.
\end{proof}  

We are ready to prove the correctness of the DAC algorithm (Algorithm \ref{Algo:Propose_dispose_algo_simple_case}).
\medskip

\begin{theorem}\label{teo:simple_propose_dispose_algo_is_correct}
The allocation $\pi$, output of the DAC algorithm, is $\varepsilon$-pairwise stable with respect to the family $\mathcal{C}$.
\end{theorem}

\begin{proof}
Let $\pi := (\mu,\vec{x},\vec{y})$ be the output of Algorithm \ref{Algo:Propose_dispose_algo_simple_case} and suppose it is not $\varepsilon$-pairwise stable. Let $(d,h)$ be an $\varepsilon$-blocking pair of $\pi$, and suppose, without loss of generality, that $h \neq \mu(d)$. Let $T$ be the last iteration at which $d$ proposed. In particular, at time $T$, doctor $d$ proposed to $\mu(d)$ and not to $h$, and for any posterior proposal to $\mu(d)$, $d$ won the competition. Since $d$ won all the posterior competitions, in particular, by \Cref{prop:man_proposer_guarantees_improvement_of_epsilon,prop:winner_man_plays_feasible_payoff},
$$f_d(\pi) \geq \max\{f_{d,h'}(\bar{x},\bar{y}) : g_{d,h'}(\bar{x},\bar{y}) \geq g_{h'}(\pi) + \varepsilon , h' \in H_0, (\bar{x},\bar{y}) \in C_{d,h'} \},$$
as $f_d(\pi)$ cannot be lower than any of the reservation prices computed by $d$ during each of her competitions. Since $(d,h)$ is an $\varepsilon$-blocking pair, there exists $(\bar{x},\bar{y}) \in C_{d,h}$ such that $f_{d,h}(\bar{x},\bar{y}) \bi f_d(\pi)+ \varepsilon$ and $g_{d,h}(\bar{x},\bar{y}) \bi g_h(\pi) + \varepsilon$. Then
$$f_d(\pi) \sm \max\{f_{d,h'}(\bar{x},\bar{y}) : g_{d,h'}(\bar{x},\bar{y}) \geq g_{h'}(\pi) + \varepsilon , h' \in H_0, (\bar{x},\bar{y}) \in C_{d,h'} \},$$
which is a contradiction.
\end{proof}

From the existence of $\varepsilon$-pairwise stable allocations, we are finally able to prove the existence of $0$-pairwise stable allocations (\Cref{teo:existence_ext_stable_alloc_general_case}) passing through the compactness of the sets in $\mathcal{C}$, the continuity of payoff functions, and the finiteness of players. 

\begin{proof}{\textbf{$0$-Pairwise stable allocations existence (\Cref{teo:existence_ext_stable_alloc_general_case}}).}
Consider $\varepsilon \bi 0$. Let $\pi_{\varepsilon} := (\mu_{\varepsilon}, \vec{x}_{\varepsilon}, \vec{y}_{\varepsilon})$ be the output of the DAC algorithm (Algorithm \ref{Algo:Propose_dispose_algo_simple_case}). Thus, $\pi_{\varepsilon}$ is an $\varepsilon$-pairwise stable allocation with respect to the family $\mathcal{C}$ (\Cref{teo:simple_propose_dispose_algo_is_correct}). Consider a sequence of these profiles $(\pi_{\varepsilon})_{\varepsilon}$ with $\varepsilon$ going to $0$, and a subsequence $(\pi_{\varepsilon_k})_k$ such that $(\vec{x}_{\varepsilon_k}, \vec{y}_{\varepsilon_k})_k$ converges to a fixed strategy profile $(\vec{x},\vec{y})$, which exists as the sets $(C_{d,h} : (d,h) \in D \times H)$ are compact sets.

Since there is a finite number of possible matchings, consider a subsubsequence $(\pi_{\varepsilon_{k_l}})_l$ such that $\mu_{\varepsilon_{k_l}} = \mu, \forall l \in \mathbb{N}$, with $\mu$ a fixed matching. As $(\vec{x}_{k_l}, \vec{y}_{k_l}) \to (\vec{x},\vec{y})$ when $l \to \infty$, the sequence $\pi_{k_l}$ converges to $ \pi := (\mu,\vec{x},\vec{y})$, with $\mu$ a matching and $(\vec{x},\vec{y})$ a strategy profile. Moreover, as $\varepsilon_{k_l}$ goes to $0$, as for each $l$ the allocation $\pi_{k_l}$ is $\varepsilon_{k_l}$-pairwise stable, as the payoff functions are continuous, and as the definition of pairwise stability only includes inequalities, $\pi$ is pairwise stable.
\end{proof}  

In the following sections we will leverage \Cref{teo:existence_ext_stable_alloc_general_case} to study the models with and without commitment. 

\section{No commitment}\label{chapter:one_to_one_matching_games}

This section is devoted to presenting the model of \textit{matching games without commitment}, arguably the most natural one (from a game-theoretic perspective) to begin the study of the strategies supporting the literature stable outcomes.

The main result of this section is twofold. Matching games without commitment form a simple and rich model that requires mild assumptions about two-player games to ensure the existence of stable outcomes. However, the model is too restrictive to capture even the simplest models of matching with endogenous utilities, such as matching with transfers.

\subsection{Pairwise-Nash stability}\label{sec:external_and_internal_Nash_stability}

Suppose $d$ and $h$ agree to match and intend to play, respectively, the actions $x_d$ and $y_h$. If no specific reason forces them to respect that agreement (no binding contracts, no possibility of future punishment in repeated interaction, etc) then, for $(x_d,y_h)$ to be stable, it must constitute a Nash equilibrium of $G_{d,h}$.    
\medskip

\begin{definition}\label{def:internally_Nash_stable_matching_profile}
\textup{An allocation $\pi = (\mu,\vec{x},\vec{y})$ is \textbf{Nash stable} if for any matched couple $(d,h) \in \mu$, $(x_d,y_h)\in \text{N.E}(G_{d,h})$, i.e, $(x_d,y_h)$ is a Nash equilibrium of $G_{d,h}$.}
\end{definition}
\medskip

As players can remain single or match a better partner, a pairwise stability condition \`{a} la Gale-Shapley must also be satisfied for an allocation to be stable.
\medskip

\begin{definition}\label{def:externally_Nash_stable_matching_profile}
\textup{An individually rational allocation $\pi = (\mu,\vec{x},\vec{y})$ is \textbf{pairwise-Nash stable} if it is pairwise stable with respect to the family
$$\mathcal{C} := \{\mathrm{N.E}(G_{d,h}), (d,h)\in D\times H \},$$
where $\mathrm{N.E}(G_{d,h})$ is the set of Nash equilibria of the game $G_{d,h}$.}
\end{definition}
\medskip

Pairwise stability asks for the non-existence of (Nash-)\textit{blocking pairs}, i.e., there is no pair $(d,h) \in D \times H$, that can be paired and play a Nash equilibrium in their game strictly improving their payoffs in $\pi$.
\medskip


\begin{theorem}\label{teo:Existence_Nash_Stable}
If for any couple $(d,h)$ the set of Nash equilibria of the game $G_{d,h}$ is non-empty and compact,\footnote{Whenever a game $G_{d,h}$ has convex and compacts strategy sets, and utility functions are own-quasi concave and continuous, or discontinuous but better-reply-secure \cite{reny1999existence}, the set of Nash equilibria is non-empty and compact.} and the payoff functions are continuous, then the set of Nash-pairwise stable allocations is also non-empty and compact.
\end{theorem}
\medskip

\Cref{teo:Existence_Nash_Stable} is a corollary of \Cref{teo:existence_ext_stable_alloc_general_case}. Indeed, pairwise-Nash stability is guaranteed by \Cref{teo:existence_ext_stable_alloc_general_case}, while Nash stability holds as all couples are restricted to play only Nash equilibria. 
We end this section by studying the Nash and pairwise-Nash stable allocations of our leading examples.
\medskip

\noindent\textbf{\Cref{ex:prisoners_dilemma_example}}. The only Nash equilibrium of the prisoners' dilemma is to play $(B,B)$. Then, as both players are better off being single than playing the equilibrium, the only Nash-pairwise stable allocation is the one in which players do not match.\qed
\medskip


\noindent\textbf{\Cref{ex:one_couple_transfer_game_example}}.
The only Nash equilibrium of the constant-sum game is $x = y = 0$, as any positive transfer $x >0$ (resp. $y >0$) is a strictly dominated strategy for $d$ (resp. for $h$). Thus, for any Nash stable allocation $\pi$ in which the players are matched, their payoffs are $f_d(\pi) = 10\delta$ and $g_h(\pi) = 0$. As this allocation is not individually rational for $h$, the only Nash-pairwise stable allocation is the one in which agents do not match.\qed
\medskip

\noindent\textbf{\Cref{ex:auction}}. Similarly to the previous example, the only Nash equilibrium is to pay $0$ for the good. However, as the seller has a null IRP, it is not individually rational for her to sell the good without receiving a positive payment. Therefore, the only Nash-pairwise stable allocation is the one in which nobody buys the good. \qed
\medskip

The Nash-pairwise stable allocations found in these examples are not Pareto-optimal. In \Cref{ex:prisoners_dilemma_example} players can match, cooperate, and end up both better off. In \Cref{ex:one_couple_transfer_game_example}, $d$ can propose $x \geq \delta$ to $h$ and both agents end up better off. In \Cref{ex:auction}, the buyer with the higher valuation can pay any price between her valuation and the one of the seller, and both, buyer and seller, end up better off. These deviations are credible and stable only if both agents believe that the other one will honor her promise (they can commit). 

Commitment is a classical assumption in matching models with transfers \citep{shapley1971assignment,rochford1984symmetrically,demange1985strategy,roth1988interior,bennett1988consistent} and matching with contracts \citep{hatfield2005matching}. Therefore, the rest of this article is devoted to study the model of \textit{matching games under commitment}.

\section{Commitment}\label{chapter:external_stability_one_to_one_commitment}

\subsection{Pairwise stability}\label{sec:external_stability_existence}

Suppose that matched partners within a couple can commit to playing a specific action profile. This allows them to enlarge their options well beyond their set of Nash equilibria. This leads to the following stronger stability notion.
\medskip

\begin{definition}\label{def:externally_stable_matching_profile}
\textup{An individually rational allocation $\pi=(\mu,\vec{x},\vec{y})$ is \textbf{pairwise stable} if it is pairwise stable with respect to the familiy
$$\mathcal{C} = \{X_d\times Y_h, (d,h)\in D\times H\},$$}
i.e., with respect to the whole set of strategy profiles.
\end{definition}  
\medskip


Pairwise stability, compared to pairwise-Nash stability (\Cref{def:externally_Nash_stable_matching_profile}), allows the players to choose any feasible strategy profile without being restricted to a Nash equilibrium (or to a non-dominated strategy). Let us see the impact of such a change in our leading examples. 
\medskip

\noindent\textbf{\Cref{ex:prisoners_dilemma_example}.} Recall the prisoners' dilemma matching game in which two agents $d$ and $h$ with positive IRPs $\delta$ can match and play,
\begin{table}[H]
\centering
\begin{tabular}{cccc}
     & \multicolumn{3}{c}{Agent h}\\\noalign{\vskip 2mm}
     \multirow{3}{*}{Agent d  } & & Cooperate & Betray \\
     \cmidrule{2-4}
     & Cooperate & $2\delta, 2\delta$ & $-\delta,3\delta$ \\ 
     \cmidrule{2-4}
     & Betray & $3\delta,-\delta$ & $0,0$
\end{tabular}
\hspace{1.5cm}
\end{table}

Matching and playing the Nash equilibrium of the game is not individually rational. However, notice that not matching is not pairwise stable as matching and cooperating blocks it. Suppose that players get matched and play a mixed strategy $(x,y) = ((p,1-p),(q,1-q))$. Their expected payoffs are,
\begin{align*}
    f_{d,h}(x,y) &= 2\delta pq + 3\delta (1-p)q - \delta p (1-q) = \delta(3q-p), \\
    g_{d,h}(x,y) &= 2\delta pq + 3\delta p(1-q) - \delta (1-p)q = \delta(3p-q).
\end{align*}
Imposing individual rationality we get a continuum of solutions given by,
\begin{align*}
    &q \in \left[\frac{1+p}{3}, 3p-1\right],\quad p,q \in [0,1],
\end{align*}
and represented by the blue region in \Cref{fig:prisoners_dilemma}. $\square$
\begin{figure}[H]
    \centering
    \includegraphics[scale = 0.7]{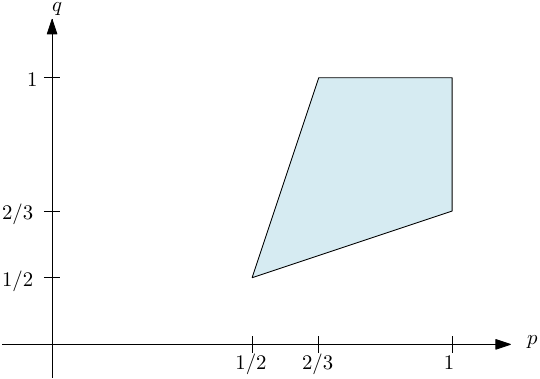}
    \caption{Stable allocations \Cref{ex:prisoners_dilemma_example}}
    \label{fig:prisoners_dilemma}
\end{figure}


\noindent\textbf{\Cref{ex:one_couple_transfer_game_example}.} Recall the transfer matching game in which two agents $d$ and $h$ with positive IRPs $\delta$ can match and play,
\begin{align*}
&G = (\mathbb{R}_+,\mathbb{R}_+,f,g), \text{ such that for any } x,y \geq 0,\\
&f_{d,h}(x,y) = 10\delta - x + y,\\
&g_{d,h}(x,y) = x - y.
\end{align*}

As we already observed, not matching, or matching and playing null transfers are not stable.  Imposing individual rationality we find that any allocation in which agents match and the transfer profile satisfies $x - y \in [\delta, 9\delta]$, is pairwise stable. This is exactly the prediction of Shapley-Shubik's and Demange-Gale's models. \qed
\medskip

\noindent\textbf{\Cref{ex:auction}.} Recall the auction example in which a seller sells an indivisible good to a set of $n$ buyers. Buyers have valuations for the good satisfying $v_1 > v_2 \geq ... \geq v_n$. The seller has valuation $c$ for the good, with $c \leq v_1$, and everybody has null IRP. If a couple buyer-seller $(d,h)$ is created and the monetary transfers $(x_d, y_h)$ is agreed, the item is sold from $h$ to $d$ at the price $p = x_d - y_h$. In particular, the utilities of these two players are,
\begin{align*}
    f_{d,h}(x_d,y_h) &= v_d - p = v_d - x_d + y_h,\\
    g_{d,h}(x_d,y_h) &= p - c = x_d - y_h - c,
\end{align*}
An allocation in this problem is pairwise stable if and only if the first buyer and the seller decide to trade at some price $p_1 \in [\max(c, v_2), v_1]$. Notice that no other buyer is willing to pay a price higher than $p_1$ as this is not individually rational. Due to this, the seller has no incentive in abandoning buyer $1$ and matching with somebody else. As, in addition, the allocation is individually rational for buyer 1 and the seller, the allocation is, indeed, pairwise stable. \qed
\medskip

As for the non-commitment setting, the general deferred-acceptance with competitions algorithm (Algorithm \ref{Algo:Propose_dispose_algo_simple_case}) can be used to compute $\varepsilon$-pairwise stable allocations for the model with commitment. 








The existence of pairwise stable allocations in the model with commitment is obtained by taking $C_{d,h} = \mathrm{PO}(G_{d,h})$, where $\mathrm{PO}(G_{d,h})$ is the set of Pareto-optimal strategy sets of game $G_{d,h}$, for any $d \in D$ and $h \in H$. 
\medskip

\begin{theorem}\label{teo:externally_stable_allocation_existence}
For any matching game with non-empty compact Pareto-optimal strategy sets and continuous payoff functions, there always exists a pairwise stable allocation.
\end{theorem}

\subsection{Shapley-Shubik and Demange-Gale models}\label{sec:shapley_shubik_gale_demange_models}

The assignment game \citep{shapley1971assignment} and the matching with transfers model \citep{demange1985strategy} are two of the main models of matching with endogenous utilities so we dedicate this section to formally establish the connection between their models and ours. The Assignment game consists of a housing market with buyers and sellers. Each seller has a house to sell and each buyer is interested in buying a house. A solution to this problem is a pair $(\mu,\vec{p})$, with $\mu$ a matching between sellers and buyers, and $\vec{p}$ a vector of positive monetary transfers from buyers to sellers. Each seller $h \in H$, has a cost of her house $c_h$, and each buyer $d \in D$, has a valuation $v_{d,h}$ for $h$'s house. If seller $h$ sells her house to $d$ at price $p_{d,h} \geq 0$, their payoffs are,
\begin{align*}
f_{d,h}(p_{d,h}) &= v_{d,h} - p_{d,h},\ \ g_{d,h}(p_{d,h}) = p_{d,h} - c_h, 
\end{align*}
respectively for buyer and seller. Demange and Gale generalized the problem by considering that whenever two agents $d$ and $h$ are paired, their payoffs are given by some strictly increasing and continuous payoff functions $\phi_{d,h}(t)$ for $d$, and $\psi_{d,h}(-t)$ for $h$, with $t \in \mathbb{R}$ being the net transfer from $d$ to $h$  ($t\geq 0$ means that $ d$ pays $t$ to $h$ and $t \leq 0$ means that $h$ pays $-t$ to $d$). 

This model can be mapped into a matching game in which all couples play \textit{strictly competitive games} (\Cref{def:strictly_competitive_game}). Formally, given $(d,h) \in D\times H$, consider the strategic game $G_{d,h} = (X_d,Y_h,f_{d,h},g_{d,h})$, with $X_d = Y_h = \mathbb{R}_+$, and,
\begin{align*}
f_{d,h}(x_d,y_h) &= \phi_{d,h}(y_h-x_d), \ \ g_{d,h}(x_d,y_h) = \psi_{d,h}(x_d-y_h).
\end{align*}
Although the strategy sets in the Demange-Gale matching game are not compact, as transfers are naturally bounded by players' valuation and their individually rational payoffs, the problem is easily compactified, satisfying the assumptions for the existence of pairwise stable allocations (\Cref{teo:externally_stable_allocation_existence}).

It is worth noting that the models in \cite{rochford1984symmetrically}, \cite{roth1988interior}, and \cite{bennett1988consistent} can also be mapped to matching games by considering couples bargaining over their transfers. For example, this can be done using the alternating offers bargaining game of \cite{osborne2019bargaining}, with players possessing the appropriate level of patience.
\medskip

\begin{remark}
As matched players play formally a two player game when they decide on their monetary transfers, and as they are rational, each player must be best-replying to the other player's strategy. But the unique Nash equilibrium of this strictly competitive games is $x^*=y^*=0$, e.g. there is no monetary transfer. In fact, Nash stability (\Cref{def:internally_Nash_stable_matching_profile}) is too demanding under  commitments, when (monetary) contract are enforceable. In the next section, we will show that even under commitment, a form of strategic rationality is still possible leading to a novel refinement of pairwise stability.
\end{remark}

\section{Commitment with renegotiation Proofness}\label{chapter:internal_stability_one_to_one_commitment}


As observed above, modeling the strategic behavior of our agents by requiring them to play a Nash equilibrium of their game contradicts pairwise stability and Pareto-optimality. Instead, we may ask each player to best reply to the partner's action but subject to the \textit{participation constraint} of the partner (who, otherwise, may prefer remain single or be with another partner). This Nash-like behavior is logically implied by a natural refinement of pairwise-stability called \textit{renegotiation proofness}.



Participation-constrained Nash equilibria, particular case of \textit{generalized Nash equilibria} \citep{harker1991generalized}, may not exist in all the class of games where a Nash equilibria do. Due to this, a new class of strategic games, called \textit{feasible games}, is considered. We show that large classes of well-know strategic games are feasible and admit participation-constrained Nash equilibria.


We design a \textit{renegotiation process} which, for any feasible matching game instance and starting from any pairwise stable allocation, outputs a pairwise-renegotiation-proof stable allocation if it converges. Finally, we prove the convergence of the algorithm for many feasible games by designing game-dependent oracles. Additional properties will rise for matching games in which couples play infinitely repeated games.

\subsection{Renegotiation proofness}\label{sec:internal_stability}

We start by introducing the following useful notation: Given a strategy profile $\vec{x} \in X_D$, and $s_d \in X_d$ a particular strategy of some fixed doctor $d \in D$, we write $(\vec{x}_{-d},s_d)$ to the strategy profile obtained when replacing $x_d$, the strategy of doctor $d$ in $\vec{x}$, by $s_d$.
\medskip

\begin{definition}\label{def:internally_stable_matching_profile}
\textup{A \textbf{pairwise stable} allocation $\pi=(\mu,\vec{x},\vec{y})$ is \textbf{renegotiation proof} if for any couple $(d,h)\in \mu$ and any potential deviation $(s_d,t_h)\in X_d \times Y_h$, it holds,
\begin{itemize}
\item[1.] If $f_{d,h}(s_d,y_h) > f_d(\pi)$ then, $(\mu,(\vec{x}_{-d},s_d),\vec{y})$ is not pairwise stable,
\item[2.] If $g_{d,h}(x_d,t_h) > g_h(\pi)$ then, $(\mu,\vec{x},(\vec{y}_{-h},t_h))$ is not pairwise stable.
\end{itemize}}
\end{definition}  

Condition 1 means that no matched doctor $d$ can profitably deviate in actions without breaking the pairwise stability of the allocation or, in other words, without creating a blocking pair or violating the IRP of the partner. In terms of contract theory, an allocation will be renegotiation proof if any alternative proposal in the advantage of the deviating agent from the agreed action profile is rejected by the partner as she will prefer her outside option (becoming single or changing of partner). 

Renegotiation proofness implies a Nash equilibrium condition subject to a participation constraint: players must play strategies that maximize their payoff under the constraint that the partners still agree to sign the contract (\Cref{teo:int_stability_is_equivalent_to_CNEs}).


An interesting family of strategic games in which the output of the deferred-acceptance with competitions algorithm is not only pairwise stable but also renegotiation proof is the class of \textbf{common interest games}, in which $f_{d,h} = g_{d,h}$ for any couple $(d,h)\in D\times H$, as every time a doctor maximizes her payoff, she also does it for her partner hospital. In general, however, the constructed pairwise stable allocation will not be renegotiation proof. We will see in the following sections how, under some assumptions on the family of strategic games, one can, from any pairwise stable allocation, construct a pairwise stable and renegotiation proof one. Let us see the impact on renegotiation proofness in our leading examples.
\medskip

\noindent\textbf{\Cref{ex:prisoners_dilemma_example}.} Recall the prisoners' dilemma matching game example with two players, both with positive IRP $\delta$, and payoff matrix,
\begin{table}[H]
\centering
\begin{tabular}{cccc}
     & \multicolumn{3}{c}{Agent h}\\\noalign{\vskip 2mm}
     \multirow{3}{*}{Agent d  } & & Cooperate & Betray \\
     \cmidrule{2-4}
     & Cooperate & $2\delta, 2\delta$ & $-\delta,3\delta$ \\ 
     \cmidrule{2-4}
     & Betray & $3\delta,-\delta$ & $0,0$
\end{tabular}
\hspace{1.5cm}
\end{table}
\noindent We have found that any allocation in which agents match and play a mixed strategy $(x,y) = ((p,1-p),(q,1-q))$ satisfying,
\begin{align}\label{eq:externally_stable_condition_prisoners_dilemma_example}
    q \in \left[\frac{1+p}{3},3p - 1 \right], \text{ such that } p,q \in [0,1],
\end{align}
is pairwise stable. An allocation is renegotiation proof if any profitable deviation in actions of a player breaks its pairwise stability. Recall that the expected payoffs of the agents are $f_d(x,y) = \delta(3q-p)$ and $g_h(x,y) = \delta(3p-q)$. 

Since there are no other couples, the only way to break the pairwise stability of the allocation is to decrease the payoff of the partner below $\delta$ by deviating. Imposing the agents' payoffs equal to $\delta$, we obtain the strategy profile $$(x^*,y^*) = \bigl((1/2,1/2),(1/2,1/2)\bigr),$$
(the bottom-left vertex of the light blue region in \Cref{fig:prisoners_dilemma}). Let us prove that matching and playing $(x^*,y^*)$ is pairwise stable and renegotiation proof. First of all, since $(x^*,y^*)$ satisfies \Cref{eq:externally_stable_condition_prisoners_dilemma_example} the allocation is pairwise stable. Take $\varepsilon \bi 0$ and suppose that player $d$ deviates to $x' = (1/2 - \varepsilon, 1/2 + \varepsilon)$. Then,
\begin{align*}
    f_d(x',y^*) = \delta\left(3\cdot\frac{1}{2} - \frac{1}{2} + \varepsilon\right) = f_d(x^*,y^*) + \delta\varepsilon \bi f_d(x^*,y^*),
\end{align*}
so $d$ has a profitable deviation. In general, agents profitably deviate if they decrease the probability of cooperating. It follows that,
\begin{align*}
    g_h(x',y^*) = \delta\left(3\cdot\left(\frac{1}{2} - \varepsilon\right) - \frac{1}{2}\right) = g(x^*,y^*) - 3\delta\varepsilon = \delta(1 - 3\varepsilon) \sm \delta.
\end{align*}
Thus, any profitable deviation of player $d$ from $(x^*,y^*)$ decreases the payoff of player $h$ below $\delta$ violating its individual rationality. By symmetry, the same holds for any profitable deviation of player $h$ from $(x^*,y^*)$. Therefore, matching together and playing $(x^*,y^*)$ is pairwise stable and renegotiation proof. It is easy to check that this is, indeed, the only pairwise stable and renegotiation proof allocation.\qed
\medskip

\noindent\textbf{\Cref{ex:one_couple_transfer_game_example}.} Recall the transfer matching game with two agents, both with positive IRPs $\delta$, who play the constant-sum game,
\begin{align*}
&G = (\mathbb{R}_+,\mathbb{R}_+,f,g), \text{ such that for any } x,y \geq 0,\\
&f(x,y) = 10\delta - x + y,\\
&g(x,y) = x - y.
\end{align*}
\noindent Pairwise stability implies that players match and make a transfer profile $(x,y)$ such that $\delta\leq x - y \leq 9 \delta$. If $x - y > \delta$, decreasing slightly $x$ increases $d$'s payoff without violating $h$'s individual rationality. Thus, renegotiation proofness implies that $x - y = \delta$. If $x - y = \delta$ and $y > 0$, decreasing slightly $y$ increases $h$'s payoff without inciting $d$ to leave the couple. Thus, an allocation is pairwise stable and renegotiation proof if and only if $d$ and $h$ agree to match, $x = \delta$ and $y = 0$. We recover the \textbf{competitive price equilibrium} of Shapley and Shubik. \qed
\medskip

\noindent\textbf{\Cref{ex:auction}.} Recall the auction example in which a seller sells an indivisible good to a set of $n$ buyers. Buyers have valuations $v_1 > v_2 \geq ... \geq v_n$ for the good, the seller has valuation $c$ for the good, everybody has null IRP, and 
recall we have assumed $v_1 \geq c$. We have found that pairwise stability implies that the good is sold to buyer $1$ for a price $$p_1 \in [\max(c,v_2),v_1].$$ 
In this continuum, the unique renegotiation proof allocation is $p_1 = \max(c,v_2)$ (we recover the \textbf{2nd price auction outcome}). Indeed, if $p_1 \bi \max(c,v_2)$, decreasing slightly $p_1$ increases buyer $1$'s payoff without breaking the pairwise stability as the seller still prefers to be matched with her. If $p_1 = \max(c,v_2)$ and $c \bi v_2$, decreasing slightly $p_1$ turns seller's payoff negative making her to prefer to become unmatched. Similarly, if $v_2\bi c$, decreasing slightly $p_1$ creates a blocking pair as buyer $2$ and the seller could match together and increase strictly their payoffs by trading the good at some price between $p_1$ and $v_2$. Remark this solution is not Nash stable, as the selected buyer performs a positive payment, i.e., a dominated strategy in her transfer game.\qed

\subsection{Constrained Nash equilibria}\label{sec:constrained_Nash_eq}

In this section the concept of constrained Nash equilibrium (CNE) is introduced and expressed through a quasi-variational inequality formulation. The renegotiation proofness characterization by constrained Nash equilibria is delegated to \Cref{sec:strategy_profiles_modification_algorithm}.

Constrained Nash equilibria are defined for any two-player game and do not depend on the whole matching game $\Gamma$ considered. Because of this, during \Cref{sec:constrained_Nash_eq,sec:feasible_games} we will only consider one strategic game $G = (X,Y,f,g)$, where $X,Y$ are compact strategy sets and $f,g$ are continuous payoff functions. We endow each player with an \textit{outside option}, $f_0,g_0 \in \mathbb{R}$, respectively, representing the payoffs the players can obtain outside of the game $G$.
\medskip


\begin{definition}\label{def:feasible_contract}
\textup{A \textit{strategy profile} $(x,y) \in X \times Y$ is $(f_0,g_0)$-\textbf{feasible} if $f(x,y) \geq f_0$ and $g(x,y) \geq g_0$.}
\end{definition}
\medskip

In words, a strategy profile is feasible for the players, for a given pair of outside options, if they can achieve at least these payoffs when playing that strategy profile.
\medskip

\begin{definition}\label{def:constrained_Nash_eq}
\textup{A $(f_0,g_0)$-feasible strategy profile $(x',y')$ is a $(f_0,g_0)$-\textbf{constrained Nash equilibrium} (CNE) if it satisfies: 
\begin{align}
\begin{split}\label{eq:constrained_Nash_eq}
f(x',y') &= \max\{f(x,y') : g(x,y')\geq g_0, x \in X\},\\
g(x',y') &= \max\{g(x',y) : f(x',y)\geq f_0, y \in Y\}.
\end{split}
\end{align}
We denote the set of $(f_0,g_0)$-constrained Nash equilibria as CNE($f_0,g_0$).}
\end{definition}
\medskip

A strategy profile satisfies Equations (\ref{eq:constrained_Nash_eq}) if any player's profitable deviation decreases the partner's payoff below her outside option. Equations (\ref{eq:constrained_Nash_eq}) can be written as a quasi-variational inequality (QVI) \citep{facchinei2007finite,harker1991generalized,noor1988quasi} with point-to-set mappings that may fail to be lower semi-continuous. Due to this, the existence of CNE is not always guaranteed, as the next example shows.
\medskip

\begin{example}\label{ex:non_feasible_game}
\textup{Let $G$ be the following finite game played in mixed strategies.\footnote{We want to thank Eilon Solan for having suggested this example.} 
\begin{table}[H]
    \centering
    \begin{tabular}{c|ccc}
    & L & M & R \\
    \hline
    \ T \ & 2,1 & -10,-10 & 3,0 \\
    \hline
    M & 3,0 & 2,1 & -10,-10 \\
    \hline
    B & -10,-10 & 3,0 & 2,1
    \end{tabular}
    \label{tab:my_label}
\end{table}
Game $G$ has only one Nash equilibrium, completely mixed, with payoffs $-5/3$ and $-3$ for players $1$ and $2$, respectively. Take null outside options, i.e., $f_0 = g_0 = 0$. The set of $(f_0,g_0)$-feasible strategy profiles is non-empty as, for example, $(T,L)$, $(M,L)$, $(M,M)$, $(B,M)$, $(T,R)$, and $(B,R)$ belong to it. However, there is no $(f_0,g_0)$-constrained Nash equilibrium. Let us prove this formally. Let $Z_{1,2}$ be the set of feasible contracts for players $1$ and $2$ and $$(x,y) = ((x_1,x_2,x_3),(y_1,y_2,y_3)),$$ be a mixed strategy profile that achieves some payoff $(f',g')$, in which each coordinate corresponds to playing Top, Medium, and Bottom, respectively for player $1$, and Left, Medium, and Right, respectively for player $2$. }

\textup{If any of the players plays a pure strategy, the other player can improve her payoff by a unilateral deviation, still respecting the outside option of the first player. For example, consider that player $1$ plays $x = (1,0,0)$, i.e., she plays Top. Then, player $2$ can deviate to play Left with probability $1$, increasing her payoff and still giving to player $1$ a positive payoff, that is, respecting her outside option. In the same way, if player $2$ plays Left, player $1$ can deviate to play Medium. Thus, no pure strategy can be a constrained equilibrium.}

\textup{The same holds if any of the players play a mixed strategy without full support. Consider that player $1$ plays $(x_1,x_2,0)$. Then, player $2$ can deviate and play Left with probability $1$ if $x_1$ is large enough, or a mixed strategy mixing only Left and Medium if $x_2$ is large enough. In any of the two cases, players converge to play pure strategies, which we already saw cannot be a constrained Nash equilibrium.}

\textup{Consider that both players play mixed strategies with full support. The Nash equilibrium of the game not belonging to $Z_{1,2}$, it cannot be the case that players play $(1/3,1/3,1/3)$. Without loss of generality, assume that $x_1 \bi 1/3 \geq x_2$. The expected payoff of player $2$ is given by,
\begin{align*}
    g' &= g(x,y) = y_1(12x_1 + 21x_2 - 11) + y_2(-9x_1 + 12x_2 - 1) + (1-x_1 - 11x_2).
\end{align*}
It holds that $x_1,x_2,y_1,y_2$ are strictly positive and $x_1 + x_2 \sm 1$, $y_1 + y_2 \sm 1$, since players have full support. Then, $-9x_1 + 12x_2 - 1 \sm 0$, so player $2$ can deviate and increase her own payoff by decreasing $y_2$. The expected payoff of player $1$ is,
\begin{align*}
    f' &= f(x,y) = y_1(11x_1 + 25x_2 - 12) + y_2(-14x_1 + 11x_2 + 1) + (x_1 - 12x_2 + 2).
\end{align*}
It holds $-14x_1 + 11x_2 + 1 \sm 0$, so player $1$ increases her payoff if $y_2$ decreases as well. Therefore, there exists a profitable deviation for player $2$ that still guarantees to player $1$ her outside option. Intuitively, since player $1$ is more likely to play Top, it makes sense that both players improve their payoff if player $2$ decreases the probability of playing Medium, so they avoid getting $-10$. We conclude that $(f',g')$ is not the payoff of a constrained equilibrium payoff, and therefore, the game $G$ is not feasible.
}\qed
\end{example}

\subsection{Feasible games}\label{sec:feasible_games}

Constrained Nash equilibria are the key to obtaining renegotiation proof allocations. Since CNE are not guaranteed to exist as the previous section explained, we will consider the following class of games.
\medskip

\begin{definition}\label{def:feasible_game}
\textup{A two-person game $G$ is \textbf{feasible} if for any pair of outside options $(f_0,g_0) \in \mathbb{R}^2$, which admits at least one $(f_0,g_0)$-feasible strategy profile, there exists a $(f_0,g_0)$-CNE.}
\end{definition}
\medskip

\textit{Feasibility} is a necessary condition for the existence of pairwise stable and renegotiation proof allocations. Let us illustrate it with an example.
\medskip

\noindent\textbf{\Cref{ex:one_couple_transfer_game_example}.} Recall the transfer matching game example in which two agents $d$ and $h$ with positive IRPs $\delta$ can match and play,
\begin{align*}
&G = (\mathbb{R}_+,\mathbb{R}_+,f,g), \text{ such that for any } x,y \geq 0,\\
&f_{d,h}(x,y) = 10\delta - x + y,\\
&g_{d,h}(x,y) = x - y.
\end{align*}
Taking $(f_0,g_0) = (\delta,\delta)$, any transfer profile $(x,y)$ satisfying $x - y \in (\delta,9\delta)$ is $(f_0,g_0)$-feasible. The $(f_0,g_0)$-CNE of this game corresponds to the renegotiation proof allocation found before: $x = \delta$ and $y = 0$. \qed
\medskip


As not all two-player games are feasible, we dedicate the rest of this section to prove the following theorem, showing the richness of this class of games.
\medskip

\begin{theorem}\label{teo:Class_of_feasible_games}
The class of feasible games includes,
\begin{itemize}
\item[1.] Constant-sum games with a value,
\item[2.] Strictly competitive games with an equilibrium,
\item[3.] Potential games,
\item[4.] Infinitely repeated games.
\end{itemize}
\end{theorem}

The proof that a game is feasible relies on the characteristic of the game. In other words, the proof is game-dependent. Therefore, we prove \Cref{teo:Class_of_feasible_games} in several subsections. In addition, we briefly recall each of the games mentioned in the theorem. 

Infinitely repeated games, besides being feasible, will satisfy many interesting properties. We recall these games and study their feasibility in \Cref{sec:infinitely_repeated_games_are_feasible}. The mentioned properties are delegated to \Cref{sec:the_strongest_stability}.

\subsubsection{Zero-sum games with a value are feasible}\label{sec:zero_sum_games_are_feasible} 

A two-person game $G = (X,Y,f,g)$ is a zero-sum game if players' payoff functions satisfy $f(\cdot, \cdot) = - g(\cdot, \cdot) =: u(\cdot, \cdot)$. 

\begin{definition}
\textup{Given $(f_0,g_0)$ player's outside options, with $f_0 \leq g_0$, a strategy profile $(x,y) \in X \times Y$ is \textbf{feasible} if and only if it satisfies $$f_0 \leq u(x,y) \leq g_0.$$
A feasible payoff profile $(x',y')$ is a $(f_0,g_0)$-\textbf{CNE} if for any $(x,y) \in X \times Y$, it holds that, if $u(x,y') \bi u(x',y')$ then, $u(x,y') \bi g_0$ and, if $u(x',y) \sm u(x',y')$ then, $u(x',y) \sm f_0$.}
\end{definition}

\begin{proof}{\textbf{Zero-sum games with a value are feasible}.} Let $G = (X,Y,u)$ be a zero-sum game, with $X,Y$ compact convex subsets of topological vector spaces and $u$ separately continuous. Suppose the game $G$ has a value $w$ and by continuity of $u$ and compactness of $X$ and $Y$, players have optimal strategies $(x^*,y^*)$. Let $(x',y')$ be a feasible contract ($f_0 \leq u(x',y') \leq g_0$). The analysis is split into three cases.
\medskip

\noindent\textbf{Case 1.} $f_0 \leq w \leq g_0$. The optimal contract $(x^*,y^*)$ is feasible. Since $(x^*,y^*)$ is a Nash equilibrium, it is a $(f_0,g_0)$-constrained Nash equilibrium.
\medskip

\noindent\textbf{Case 2.} $w  \sm f_0  \sm g_0$. Consider the set $A(f_0) := \{x \in X: \exists y \in Y, u(x,y) \geq f_0\}$. Since $(x',y')$ is a feasible contract, $A(f_0)$ is non-empty. Consider the optimization problem
\begin{align*}\label{eq:problem_P_zero_sum_game}\tag{P}
  \sup\left[ \inf \{ u(x,y) : u(x,y) \geq f_0, y \in Y\} : x \in A(f_0)\right].
\end{align*}
For a given $x_0 \in A(f_0)$, the set $\{y \in Y: g(x_0,y) \geq f_0\}$ is bounded and so, there exists an infimum $y_0(x_0)$. Thus, as the set $A(f_0)$ is also bounded, there exists a supremum $x_0$. Let $(x_0,y_0(x_0))$ be the pair supremum-infimum solution of (\ref{eq:problem_P_zero_sum_game}). It holds that $u(x_0,y_0(x_0)) \geq f_0$ by construction. Suppose that $u(x_0,y_0(x_0)) \bi f_0$. Since $w  \sm f_0$, it holds $w  \sm f_0  \sm u(x_0,y_0(x_0))$. Considering the optimal contract $(x^*,y^*)$, it holds $u(x_0,y^*) \leq u(x^*,y^*) = w  \sm f_0  \sm u(x_0,y_0(x_0))$. By continuity of the function $u(x_0,\cdot)$, there exists $\lambda \in (0,1)$ such that $u(x_0, y_{\lambda} ) = f_0$, with $y_{\lambda} = \lambda y^* + (1-\lambda)y_0(x_0) \in Y$. This contradicts the fact that $(x_0, y(x_0))$ is the solution to (\ref{eq:problem_P_zero_sum_game}). Thus, $u(x_0,y(x_0)) = f_0$. If this strategy profile is a constrained Nash equilibrium, the study of the second case is done. If not, consider $y_{t}\in Y$ as the convex combination between $y(x_0)$ and $y^*$ with $t$ computed by,
\begin{align}\label{eq:computation_of_t}
  t := \sup\{\tau \in [0,1] : y_{\tau} := (1-{\tau})y(x_0) + {\tau} y^* \text{ and } \exists x_{\tau} \in X, u(x_{\tau},y_{\tau}) = f_0 \}.  
\end{align}
$t$ exists as for ${\tau} = 0$, there exists $x_0$ such that $u(x_0,y(x_0)) = f_0$. In addition, $y_{t} \neq y^*$, since the contract $(x^*,y^*)$ is a saddle point, $u(x^*, y^*) = w  \sm f_0$ and any deviation of player $1$ decreases the payoff. Notice that any profitable deviation of player $2$ decreases the payoff below $f_0$, as $u(x_{t},y_{t}) = f_0$. Suppose there exists $\hat{x} \in X$ such that $f_0 = u(x_{t}, y_{t}) \sm u(\hat{x},y_{t}) \leq g_0$. As a summary, it holds: $u(x^*,y^*) = w  \sm f_0 =  u(x_{t}, y_{t}) \sm u(\hat{x}, y_{t})$ with $y_{t}$ in the interval $(y(x_0), y^*)$. Once again, by the continuity of $u$ and the convexity of $X \times Y$, there exists an element $z$ in the interval $(y_{t},y^*)$ and some $x_z \in X$, such that $u(x_z,z) = f_0$, contradicting the definition of $t$ in (\ref{eq:computation_of_t}). Thus, $(x_{t},y_{t})$ is a constrained Nash equilibrium.
\medskip

\noindent\textbf{Case 3.} $f_0  \sm g_0  \sm w$. Analogous to case 2.
\end{proof} 

Zero-sum games with a value are not only feasible but for any $(f_0,g_0)$-CNE $(x,y)$, it is satisfied,
$$u(x,y) = \text{median}(f_0,w,g_0).$$
This equation will be particularly useful later. Let us illustrate the previous proof with the following example. 
\medskip

\begin{example}\label{ex:matching_pennies}
\textup{Consider the \textit{matching pennies} game (\Cref{tab:matching_pennies}) played in mixed strategies. Matching pennies is a zero-sum game of value $w = 0$ and optimal mixed strategies $(x^*, y^*) = (1/2, 1/2)$. \Cref{fig:matching_pennies} represents the payoff function $u$ (in red), the value $w$ (black dot), two pair of outside options $(f_0,g_0)$ and $(f'_0,g'_0)$ for the players, and two strategy profile payoffs $u(x,y)$ (blue dot) and $u(x',y')$ (green dot). It holds $$w \sm f_0 \sm g_0 \text{ and } f'_0 \sm g'_0 \sm w,$$
being, respectively, the second and third cases studied during the previous proof.} \noindent\textup{We can observe that $(x,y) \in$ CNE$(f_0,g_0)$ and $(x',y') \in$ CNE$(f'_0,g'_0)$. Indeed, if $w \sm f_0 \sm g_0$ holds, any profitable deviation of Player $2$ from $(x,y)$ makes the payoff of the game lower than $f_0$, so Player $1$ receives less than her outside option. Analogously, in the case $f'_0 \sm g'_0 \sm w$, if Player $1$ can deviate from $(x',y')$ increasing her payoff, Player $2$ receives less than $g'_0$.} \qed
\end{example}
\medskip

\begin{remark}
\textup{From the nature of zero-sum games, we obtain that the constrained Nash equilibrium selection provides always a \textbf{Pareto-optimal} action profile.}
\end{remark}

\textup{\begin{minipage}[h]{0.4\linewidth}
\begin{table}[H]
\begin{tabular}{cccc}
& \multicolumn{3}{c}{Player 2}\\\noalign{\vskip 2mm}
\multirow{3}{*}{Player 1 } & & A & B \\
\cmidrule{2-4}
& A & 1 & -1 \\
\cmidrule{2-4}
& B & -1 & 1
\end{tabular}
\caption{Matching pennies}
\label{tab:matching_pennies}
\end{table}
\end{minipage}\hfill
\begin{minipage}[H]{0.58\linewidth}
\begin{figure}[H]
\centering
\includegraphics[scale = 0.7]{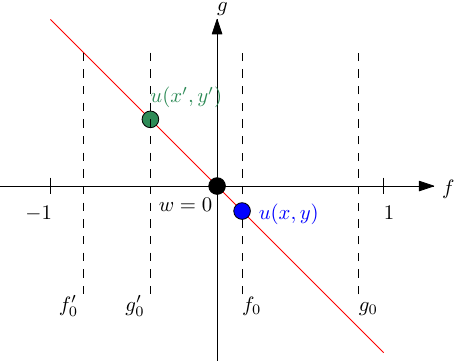}
\caption{Feasible payoffs region}
\label{fig:matching_pennies}
\end{figure}
\end{minipage}}
\medskip

The feasibility of strictly competitive games is a consequence of the one for zero-sum games. For a formal proof, please check Appendix \ref{sec:feasibility_strictly_competitive_games}.

\subsubsection{Potential games are feasible} 

A two-person game $G = (X, Y, f,g)$ is a potential game if there exists a(n exact) potential function $\phi : X \times Y \to \mathbb{R}$ such that,  $\forall x,x' \in X, y,y' \in Y$ it satisfies,
\begin{align*}
\phi(x',y') - \phi(x,y') &= f(x',y') - f(x',y), \\
\phi(x',y') - \phi(x',y) &= g(x',y') - g(x',y).
\end{align*}

\begin{proof}{\textbf{Potential games are feasible}.} Let $G= (X, Y, f,g)$ be a potential game with potential function $\phi$. Let $(f_0,g_0)$ be outside options and $Z_0$, be the set of all $(f_0,g_0)$-feasible contracts. Suppose $Z_0 \neq \emptyset$. We aim to prove that $Z_0$ includes at least one $(f_0,g_0)$-CNE. Consider 
$$(x',y') \in \argmax \{\phi(x,y) : (x,y) \in Z_0 \}.$$ 
Remark $(x',y')$ always exists as $Z_0$ is a non-empty compact set and $\phi$ is continuous. It holds that $(x',y')$ is $(f_0,g_0)$-feasible. Consider $x \in X$ such that $f(x,y') \bi f(x',y')$ and $g(x,y') \geq g_0$. In particular, 
$$f(x,y') \bi f(x',y') \geq f_0,$$
so $(x,y') \in Z_0$. Moreover, $\phi(x,y') \bi \phi(x',y')$ as $f(x,y') \bi f(x',y')$. This contradicts that $(x',y')$ belongs to the argmax of $\phi$ in $Z_0$. The same holds for player $2$. Thus, $(x',y')$ is a $(f_0,g_0)$-CNE.
\end{proof}

Unlike the previous feasible games studied, Pareto-optimality and the constrained Nash equilibria requirements are not always compatible for potential games. Consider, for example, the following prisoners' dilemma,
\begin{table}[H]
\centering
\begin{tabular}{cccc}
     & \multicolumn{3}{c}{Player 2}\\\noalign{\vskip 2mm}
     \multirow{3}{*}{Player 1  } & & Cooperate & Betray \\
     \cmidrule{2-4}
     & Cooperate & $2, 2$ & $-1,3$ \\ 
     \cmidrule{2-4}
     & Betray & $3,-1$ & $0,0$
\end{tabular}
\hspace{1.5cm}
\end{table}
Given outside options equal to $0$ for both agents, the unique CNE is the Nash equilibrium where both players Betray, which is Pareto-dominated by both player cooperating. 


 
\subsubsection{Infinitely repeated games are feasible}\label{sec:infinitely_repeated_games_are_feasible} 

Consider a two-person finite game in mixed strategies, $G = (X,Y,f,g)$, called the \textbf{stage game}, that is played in discrete time $k = \{1,...,K,...\}$ after observing the past history of plays $h_k = ((x_1,y_1),...,(x_{k-1},y_{k-1}))$. Given $K \in \mathbb{N}$, consider the $K$-stages game $G^K$ defined by the payoff functions 
\begin{align*}
f(K,\sigma_1, \sigma_2) := \frac{1}{K} \mathbb{E}_{\sigma}\left[\sum_{k = 1}^K f(x_k,y_k)\right], \quad 
g(K,\sigma_1, \sigma_2) := \frac{1}{K} \mathbb{E}_{\sigma}\left[\sum_{k = 1}^K g(x_k,y_k)\right],
\end{align*}
where $\sigma_1 : \bigcup(X\times Y)_{k=1}^{\infty} \to X$ and $\sigma_2 : \bigcup(X\times Y)_{k=1}^{\infty} \to Y$ are the players' \textit{behavioral strategies}. We define the \textbf{uniform game} $G_{\infty}$ as the game obtained by taking $K \to \infty$ in $G^K$. To state the definition of constrained Nash equilibrium for uniform games, we need some preliminary concepts.
\medskip

\begin{definition}\label{def:feasible payoffs}
\textup{We define the set of \textbf{feasible payoffs} 
$$co(f,g) := co\{(f(x,y),g(x,y)) \in \mathbb{R}^2 : (x,y) \in X \times Y\},$$
where $co$ stands for the convex envelope.}
\end{definition}
\medskip

\begin{definition}\label{def:punishment_levels}
\textup{Given the \textbf{punishment level} of players $1$ and $2$,
\begin{align}
\begin{split}\label{eq:punishment_levels}
\alpha &:= \min_{y \in Y} \max_{x \in X} f(x,y) \text{ and } \beta := \min_{x \in X} \max_{y \in Y} g(x,y),
\end{split}
\end{align}
we define the set of \textbf{uniform equilibrium payoffs} (Folk theorem of \cite{aumann1994long}) as $$E = \{(\bar{f}, \bar{g}) \in co(f,g) : \bar{f} \geq \alpha, \bar{g} \geq \beta\}.$$}
\end{definition}
\medskip

\begin{definition}\label{def:acceptable_payoffs}
\textup{Consider $f_0,g_0 \in \mathbb{R}$ outside options for player $1$ and player $2$, respectively. The set of \textbf{acceptable payoffs} is defined as 
\begin{align*}
E(f_0,g_0) := \{(\bar{f}, \bar{g}) \in co(f,g) : \bar{f} \geq f_0 \text{ and } \bar{g} \geq g_0\}.
\end{align*}}
\end{definition}

We are ready to define the constrained Nash equilibria of a uniform game.
\medskip

\begin{definition}\label{def:constrained_uniform_equilibrium}
\textup{A strategy profile $\sigma = (\sigma_1, \sigma_2)$ is called a \textbf{constrained uniform equilibrium} of $G_{\infty}$ if:
\begin{itemize}
    \item[1.] $\forall \varepsilon  \bi 0$, $\sigma$ is a $(f_0,g_0)$-$\varepsilon$-constrained equilibrium of any long enough finitely repeated game, that is, $\exists K_0, \forall K \geq K_0, \forall (\tau_1,\tau_2):$
    \begin{itemize}
        \item[(a)] If $f(K,\tau_1, \sigma_1) \bi f(K,\sigma) + \varepsilon$ then $g(K,\tau_1, \sigma_2) < g_0$,
        \item[(b)] If $g(K,\sigma_1, \tau_2) \bi g(K,\sigma) + \varepsilon$ then $f(K,\sigma_1, \tau_2) < f_0$, and,
    \end{itemize}
    \item[2.] $[(f(K,\sigma),g(K,\sigma))]_K$ has a limit $[f(\sigma),g(\sigma)]$ in $\mathbb{R}^2$ as $K$ goes to infinity, with $f(\sigma) \geq f_0$, $g(\sigma) \geq g_0$.
\end{itemize}
The set of constrained uniform equilibrium payoffs is denoted as $E^{\infty}(f_0,g_0)$.}
\end{definition}
\medskip

A uniform game will be feasible if every time that the set of acceptable payoffs is non-empty, the set of constrained uniform equilibrium payoffs is non-empty as well.
\medskip

\begin{definition}\label{def:uniform_game_feasible}
$G_{\infty}$ is \textbf{feasible} if whenever $E(f_0,g_0)$ is non-empty, $E^{\infty}(f_0,g_0)$ is non-empty as well.
\end{definition}
\medskip

By the Folk theorem of \cite{aumann1994long}, the following proposition holds.
\medskip
\begin{proposition}\label{prop:uniform_equilibria_payoff_are_constrained_uniform_payoff}
Any payoff in $E \cap E(f_0,g_0)$ can be achieved by a uniform equilibrium and so, by a constrained uniform equilibrium.
\end{proposition}
\medskip

We are ready to prove the feasibility of infinitely repeated games (\Cref{teo:Class_of_feasible_games}).

\begin{proof}{\textbf{Infinitely repeated games are feasible}.} Suppose $E(f_0,g_0)$ is non-empty. We aim to show that $E^{\infty}(f_0,g_0)$ is non-empty as well. Recall the punishment levels $\alpha$ and $\beta$ of the players (\Cref{def:punishment_levels}).  The analysis is split into four cases.
\medskip

\noindent\textbf{Case 1.} $g_0 \geq \beta$ and $f_0 \geq \alpha$. It holds that $E(f_0,g_0) \subseteq E$. By \Cref{prop:uniform_equilibria_payoff_are_constrained_uniform_payoff}, $E^{\infty}(f_0,g_0) = E(f_0,g_0)$. As $E(f_0,g_0)$ is non-empty, $E^{\infty}(f_0,g_0)$ is non-empty as well. 
\medskip

\noindent\textbf{Case 2.} $g_0  \sm \beta$ and $f_0  \sm \alpha$. It holds that $E \subset E(f_0,g_0)$. Thus, $E^{\infty}(f_0,g_0)$ contains $E$ (by \Cref{prop:uniform_equilibria_payoff_are_constrained_uniform_payoff}) and so, it is non-empty.
\medskip

\noindent\textbf{Case 3.} $g_0  \sm \beta$ and $f_0 \geq \alpha$. If $F:=E(f_0,g_0) \cap E$ is non-empty, by \Cref{prop:uniform_equilibria_payoff_are_constrained_uniform_payoff}, all elements on $F$ belong to $E^{\infty}(f_0,g_0)$. Otherwise, consider $(f',g')$ defined by
$$g' :=  \max\{\bar{g} : \exists \bar{f} \text{ s.t. } (\bar{f},\bar{g}) \in E(f_0,g_0)\},\ f' \in \{\bar{f} : (\bar{f},g') \in E(f_0,g_0)\}.$$
As $E(f_0,g_0)$ is a non-empty closed set, $(f',g')$ indeed exists and it belongs to $E(f_0,g_0)$. Consider the strategy profile $\sigma'$ in which the players follow a pure plan which yields the payoff $(f',g')$. If Player $1$ deviates, Player $2$ punishes her at the level $\alpha$, and if Player $2$ deviates, Player $1$ ignores the deviation and continues to follow the pure plan. 

Player $1$ cannot gain more than $\varepsilon$ by deviating. Indeed, if she does, Player $2$ punishes her by reducing his payoff to $\alpha$. Since $(f',g') \in E(f_0,g_0)$, it holds that $f' \geq f_0 \geq \alpha$ and so, this deviation is not profitable. For Player $2$, suppose there exists $K \in \mathbb{N}$ and $\varepsilon \bi 0$ such that she can obtain a payoff $g'' \bi g' + \varepsilon$ by deviating at stage $K$. Let $f''$ be the average payoff of Player $1$ at stage $K$ after the deviation of Player $2$. Since $(f'',g'')$ is an average payoff of the $K$-stages game, it is feasible. It cannot hold that $f'' \geq f_0$, since it would contradict the definition of $g'$, as the payoff $(f'',g'')$ would be acceptable. Thus, $f'' \sm f_0$. We conclude that $\sigma'$ is a constrained equilibrium and then, $(f',g') \in E^{\infty}(f_0,g_0)$. 
\medskip

\noindent\textbf{Case 4.} $g_0 \geq \beta$ and $f_0  \sm \alpha$. Analogous to case 3. 
\end{proof} 

\begin{remark}\label{remark:oracle_for_repeated_games_is_always_Pareto}
\textup{The constrained Nash equilibrium choice can always be done \textbf{Pareto-optimally}. Moreover, except for some non standard cases, the existence of feasible uniform equilibria holds. We study this in detail in \Cref{sec:the_strongest_stability}}.
\end{remark}

\subsection{Renegotiation process}\label{sec:strategy_profiles_modification_algorithm}

Constrained Nash equilibria capture renegotiation proofness when considering the appropriated outside options (the players' reservation payoffs\footnote{Reservation payoffs where called \textit{threat levels} by \cite{rochford1984symmetrically}.}). Let 
$$\Gamma = (D,H,(G_{d,h}: d \in D, h \in H, \underline{f},\underline{g}),$$ be a matching game. 
\medskip

\begin{definition}
\textup{Given $\pi = (\mu,\vec{x},\vec{y})$ an allocation and $(d,h) \in \mu$ an arbitrary matched couple, we define their \textbf{reservation payoffs} by,
\begin{align}
    \begin{split}\label{eq:outside_options}
     f_d^{\pi} &:= \max \{f_{d,h'}(s,t) : h' \in H_0 \setminus \{h\}, g_{d,h'}(s,t) \bi g_{h'}(\pi), (s,t) \in X_d \times Y_{h'}\},\\
    g_h^{\pi} &:= \max \{g_{d',h}(s,t) : d' \in D_0 \setminus \{d\}, f_{d',h}(s,t) \bi f_{d'}(\pi), (s,t) \in X_{d'} \times Y_h\},
    \end{split}
\end{align}
Reservation payoffs are the best payoffs that $d$ and $h$ can get outside of their couple by matching with a partner who may accept them or becoming single.}
\end{definition}
\medskip

We are finally ready to characterize the renegotiation proof allocations by constrained Nash equilibria.
\medskip

\begin{proposition}\label{teo:int_stability_is_equivalent_to_CNEs}
A pairwise stable allocation $\pi = (\mu,\vec{x},\vec{y})$ is renegotiation proof if and only if for any $(d,h) \in \mu$, $(x_d,y_h)$ is a $(f_d^{\pi}, g_h^{\pi})$-constrained Nash equilibria, where $f_d^{\pi}$ and $g_h^{\pi}$ are the agents' reservation payoffs.
\end{proposition}

\begin{proof}
Suppose that all couples play constrained Nash equilibria. Let $(d,h) \in \mu$ be an arbitrary matched couple and $(x_d,y_h)$ be their $(f_d^{\pi}, g_h^{\pi})$-CNE. Suppose there exists $s \in X_d$ such that $f_{d,h}(s, y_h) \bi f_{d,h}(x_d,y_h)$. In particular, 
$$f_{d,h}(s, y_h) \bi \max \{f_{d,h}(\ell,y_h) : g_{d,h}(\ell,y_h) \geq g_h^{\pi}, \ell \in X_d \}.$$
Thus, $g_{d,h}(s,y_h) \sm g_h^{\pi}$. Let $d'$ be the player that attains the maximum in $g_h^{\pi}$. Then, $(d',h)$ is a blocking pair of the pairwise stability of $\pi$. For $h$ the proof is analogous.

Conversely, suppose $\pi$ is renegotiation proof. Let $(d,h) \in \mu$ be an arbitrary couple and $(x_d,y_h)$ be their strategy profile. Then, for any $s \in X_d$ such that $f_{d,h}(s, y_h) \bi f_{d,h}(x_d,y_h)$, it holds that $g_{d,h}(s, y_h) \sm g_h^{\pi}$. Thus, 
$$f_{d,h}(x_d,y_h) \geq \max\{f_{d,h}(\ell,y_h) : g_{d,h}(\ell,y_h) \geq g_h^{\pi}, \ell \in X_d\}.$$
For player $h$ the proof is analogous.
\end{proof}

\Cref{teo:int_stability_is_equivalent_to_CNEs} gives the first insight into the design of an algorithm to compute renegotiation proof allocations: It is enough to modify the strategy profiles of each couple by constrained Nash equilibria under outside options equal to their reservation payoffs. A second insight comes from the following result.
\medskip

\begin{proposition}\label{teo:the_outside_options_of_a_blocking_pair_are_higher_than_their_payoff}
Let $\pi = (\mu,\vec{x},\vec{y})$ be an allocation. Then, $\pi$ is pairwise stable if and only if $f_d^{\pi}\leq f_d(\pi)$ and $g_h^{\pi} \leq g_h(\pi)$, for any $(d,h) \in D \times H$, where $f_d^{\pi}$ and $g_h^{\pi}$ are agents' reservation payoffs (Equations (\ref{eq:outside_options})).
\end{proposition}

\begin{proof}
Suppose that $\pi$ is pairwise stable and let $d \in D$ be a doctor such that $f_d^{\pi} \bi f_d(\pi)$. Thus, there exists $h \in H_0 \setminus \{\mu(d)\}$ and $(s,t) \in X_d \times Y_h$ such that 
$$g_{d,h}(s,t) \bi g_h(\pi) \text{ and } f_d^{\pi} = f_{d,h}(s,t).$$
It is clear that $(d,h)$ is a blocking pair of $\pi$, so we obtain a contradiction. The same conclusion holds if for any $h \in H$, $g_h^{\pi} \bi g_h(\pi)$.

Conversely, suppose that for any $(d,h) \in D \times H$, $f_d^{\pi} \leq f_d(\pi)$ and $g_h^{\pi} \leq g_h(\pi)$. Let $(d,h) \in D \times H$ be a blocking pair of $\pi$. Then, there exists $(s,t) \in X_d \times Y_h$ such that 
$$f_{d,h}(s,t) \bi f_d(\pi) \text{ and } g_{d,h}(s,t) \bi g_h(\pi).$$
In particular, notice that $f_d^{\pi} \geq f_{d,h}(s,t)$ and $g_h^{\pi} \geq g_{d,h}(s,t)$, as each of player can offer to the other one more than their current payoffs. We obtain a contradiction.
\end{proof}

If $\pi$ is a pairwise stable allocation, \Cref{teo:the_outside_options_of_a_blocking_pair_are_higher_than_their_payoff} implies that the reservation payoffs of all couples are never greater than their current payoffs. Therefore, if $(d,h) \in \mu$, $(x_d,y_h)$ is always $(f_d^{\pi}, g_h^{\pi})$-feasible and, if their game is feasible, there always exists a $(f_d^{\pi}, g_h^{\pi})$-constrained Nash equilibrium $(\hat{x}_d, \hat{y}_h)$. Moreover, replacing the current strategy profile with the CNE will not create blocking pairs. Thus, the pairwise stability is always preserved\footnote{\cite{rochford1984symmetrically} established a similar result within the utility space. In particular, our renegotiation process and his re-bargaining process coincide when its model is mapped to a matching game model where the players play bargaining games over the transfers with equal bargaining power.}. With this in mind, a renegotiation process is designed. It will output a pairwise stable and renegotiation proof allocation for any pairwise stable allocation used as an input. Intuitively, it will replace one by one the strategy profiles of the couples by a CNE, using at each iteration the reservation payoffs (\Cref{eq:outside_options}) as outside options.

If at any iteration a couple replaces $(x_d,y_h)$ by a $(f_d^{\pi}$, $g_h^{\pi})$-feasible Nash equilibrium, they will keep playing it during all posterior iterations. If couples cannot replace their strategy profile with a Nash equilibrium, the choice of a constrained Nash equilibrium is made by an \textbf{oracle}. Our \textit{renegotiation process} is summarized in Algorithm \ref{Algo:strategy_profiles_modification}. 

\begin{algorithm}[ht]
\SetKwInOut{Input}{input}\SetKwInOut{Output}{output}
\Input{$\pi = (\mu,\vec{x},\vec{y})$ pairwise stable allocation}
\SetInd{0.2cm}{0.2cm}

$t \longleftarrow 1, \pi(t) \longleftarrow \pi$
 
\While{True}{
\For{$(d,h) \in \mu$}{
Compute the reservation payoffs $f_d^{\pi(t)}$ and $g_h^{\pi(t)}$ (\Cref{eq:outside_options}).

Choose $(x_d^*,y_{h}^*) \in$ CNE$(f_d^{\pi(t)}, g_h^{\pi(t)})$ and set $(x_d^{t+1},y_{h}^{t+1}) \longleftarrow (x_d^*,y_{h}^*)$}
\If{$\forall (d,h) \in \mu, (x_d^{t+1},y_{h}^{t+1}) = (x_d^t,y_{h,d}^t)$}{Output $\pi(t)$}
$t \longleftarrow t+1$
}
\caption{Renegotiation process}
\label{Algo:strategy_profiles_modification}
\end{algorithm}

The convergence of Algorithm \ref{Algo:strategy_profiles_modification} does not directly hold as the reservation payoffs change at each iteration. Indeed, replacing a strategy profile with a constrained Nash equilibrium may decrease the payoff of an agent. Therefore, her reservation payoff may also change and the constrained Nash equilibrium may not be an equilibrium anymore. Nevertheless, if after replacing all strategy profiles of $\pi$, the reservation payoffs remain invariant, the current allocation is indeed renegotiation proof and the algorithm stops. 

The proof of the convergence of the renegotiation process (Algorithm \ref{Algo:strategy_profiles_modification}) is game-dependent as the choice of the oracle is different for each class of games. First, we state the proof of its correctness. Then, we state the proof that the renegotiation process converges for each of the classes of games.
\medskip

\begin{theorem}[\textbf{The renegotiation process is correct}]\label{teo:strategy_profile_modification_is_correct}
If the renegotiation process (Algorithm \ref{Algo:strategy_profiles_modification}) converges, its output is pairwise stable and renegotiation proof. 
\end{theorem}

\begin{proof}
Let $\pi = (\mu,\vec{x},\vec{y})$ be the input of Algorithm \ref{Algo:strategy_profiles_modification}. By construction, whenever the algorithm converges, the output is renegotiation proof. Concerning pairwise stability, we aim to prove that if $\pi_t$, the allocation before iteration $t$, is pairwise stable then, $\pi_{t+1}$ is pairwise stable as well. Let $(d,\mu(d))$ be a couple that changes of strategy profile at iteration $t$. Let $(x_d,y_{\mu(d)})$ be their strategy profile at iteration $t$ and $(\hat{x}_d, \hat{y}_{\mu(d)})$ at time $t+1$. Suppose there exists $(i,j)$ a blocking pair of $\pi_{t+1}$. If $i \neq d$ (and analogously if $j \neq \mu(d)$) then $f_{i, \mu(i)}(\pi_{t+1}) = f_{i,\mu(i)}(\pi_{t})$. Thus, it cannot hold that both $i \neq d$ and $j \neq \mu(d)$, otherwise the pair $(i,j)$ would also block $\pi_t$. Without loss of generality, suppose that $i = d$. In particular, $j \neq \mu(d)$ because $i$ and $j$ are not a couple. It holds,
$$g_{d,j}(s,r) \bi g_{\mu(j),j} (\pi_{t+1}) = g_{\mu(j),j} (\pi_{t}),$$
where $(s,r) \in X_d \times Y_h$ is the strategy profile used by $(b,j)$ to block $\pi$. Then, if $f_d^{\pi}$ is $d$'s reservation payoff at iteration $t$ (computed by \Cref{eq:outside_options}), it holds 
$$f_d^{\pi} \geq f_{d,j}(s,r) = f_{i,j}(s,r) \bi f_{i, \mu(i)}(\pi_{t+1}) = f_{d, \mu(d)}(\hat{x}_d, \hat{y}_{\mu(d)}).$$
This contradicts the fact that $(\hat{x}_d, \hat{y}_{\mu(d)})$ is $(f_d^{\pi},g_h^{\pi})$-feasible.
\end{proof} 

The finiteness of the algorithm for zero-sum games needs a preliminary result.
\medskip

\begin{lemma}\label{lemma:outside_options_are_monotone_in_zero_sum_games}
Let $\Gamma$ be a matching game where all strategic games are zero-sum games with a value. Let $\pi = (\mu,\vec{x},\vec{y})$ be a pairwise stable allocation and $(d,h) \in \mu$ be a matched couple. Let $w_{d,h}$ be the value of their game. Consider the sequence of reservation payoffs of $(d,h)$ denoted by $(f_d^{\pi(t)}, g_h^{\pi(t)})_t$, with $t$ being the iterations of Algorithm \ref{Algo:strategy_profiles_modification}. If there exists $t^*$ such that $w_{d,h} \leq f_d^{\pi(t)}$ (resp. $w_{d,h} \geq g_h^{\pi(t)}$), then the subsequence $(f_d^{\pi(t)})_{t\geq t^*}$ (resp. $(g_h^{\pi(t)})_{t\geq t^*}$) is non increasing (resp. non decreasing).
\end{lemma}

\begin{proof}
Suppose there exists an iteration $t$ in which $w_{d,h} \leq f_d^{\pi(t)} \leq g_h^{\pi(t)}$, so couple $(d,h)$ switches its payoff to $f_d^{\pi(t)}$ (recall that constrained Nash equilibrium payoffs are equal to the median between the value of the game and the players' reservation payoffs). Let $(\hat{x}_d, \hat{y}_h)$ be the constrained Nash equilibrium played by $(d,h)$ at iteration $t$. Since $(\hat{x}_d, \hat{y}_h)$ must be $(f_d^{{\pi(t+1)}}, g_h^{{\pi(t+1)}})$-feasible, in particular it holds $f_d^{{\pi(t+1)}} \leq f_{d,h}(\hat{x}_d, \hat{y}_h) = f_d^{\pi(t)}$. Therefore, the sequence of reservation payoffs starting from $t$ is non-increasing. 
\end{proof}

\begin{theorem}[\textbf{Convergence renegotiation process zero-sum games}]\label{teo:convergence_reneg_process_zero_sum_games}
For any oracle, the renegotiation process converges for zero-sum games with a value\footnote{The convergence result in [\cite{rochford1984symmetrically} Theorem 2] can be derived from \Cref{teo:convergence_reneg_process_zero_sum_games}. Indeed, whenever the input to our renegotiation process is the worst pairwise stable allocation for one side, the two proofs align.}.
\end{theorem}

\begin{proof} At the beginning of Algorithm \ref{Algo:strategy_profiles_modification}, each couple $(d,h)$ belongs to one of the following cases: $f_d^{\pi} \leq w_{d,h} \leq g_h^{\pi}$, $w_{d,h} \leq f_d^{\pi} \leq g_h^{\pi}$, or $f_d^{\pi} \leq g_h^{\pi} \leq w_{d,h}$. In the first case, the couple plays a Nash equilibrium and never changes it afterward. In the second case, as $f_d^{\pi}$ is non increasing for $d$ (\Cref{lemma:outside_options_are_monotone_in_zero_sum_games}) and bounded from below by $w_{d,h}$, her sequence of reservation payoffs converges. Analogously, the sequence of reservation payoffs for $h$ converges on the third case. Therefore, the renegotiation process converges.
\end{proof}  

\begin{theorem}[\textbf{Convergence renegotiation process potential games}]\label{teo:convergence_reneg_process_potential_games}
There exists an oracle for potential games such that the renegotiation process converges.
\end{theorem}

\begin{proof} 
Consider a couple $(d,h) \in \mu$ and $(\hat{x}_d^t,\hat{y}_h^t)_t$ their sequence of constrained Nash equilibria along the iterations. Since $(\hat{x}_d^{t-1},\hat{y}_h^{t-1})$ is always feasible for the following iteration (\Cref{teo:the_outside_options_of_a_blocking_pair_are_higher_than_their_payoff}), the sequence $\phi_{d,h}(\hat{x}_d^t,\hat{y}_h^t)_t$ is non-decreasing over $t$. Then, as the potential functions are continuous and the strategy sets are compact, the sequences $(\phi_{d,h}(\hat{x}_d^t,\hat{y}_h^t))_t$ are converging for any couple $(d,h)$. Thus, the renegotiation process converges.
\end{proof}

\begin{theorem}[\textbf{Convergence renegotiation process infinitely repeated ga- mes}]\label{teo:convergence_reneg_process_inf_rep_games}
There exists an oracle for infinitely repeated games such that the renegotiation process converges.
\end{theorem}

\begin{proof}
Let $\pi$ be a pairwise stable allocation, $t$ an iteration, and $(d,h) \in \mu$ a couple. Let $(f_d^{\pi(t)},g_h^{\pi(t)})$ be their reservation payoffs at iteration $t$, and consider $F_t := E \cap E(f_d^{\pi(t)},g_h^{\pi(t)})$. If $F_t$ is non-empty, there exists a $(f_d^{\pi(t)},g_h^{\pi(t)})$-feasible uniform equilibrium for $(d,h)$, so they keep playing it forever. If $F_t = \emptyset$, without loss of generality, assume that $f_d^{\pi(t)} \geq \alpha$ and $g_h^{\pi(t)}  \sm \beta$, where $\alpha$ and $\beta$ are the punishment levels of $d$ and $h$, respectively. Consider the oracle used in the proof of the feasibility of infinitely repeated games (\Cref{teo:Class_of_feasible_games}). Let $(f^t,g^t)$ be the $(f_d^{\pi(t)}, g_h^{\pi(t)})$-constrained Nash equilibrium payoff chosen at iteration $t$ by the oracle, so 
$$g^t := \max \bigl\{\bar{g} \in \mathbb{R}: \exists \bar{f} \in \mathbb{R} \text{ such that } (\bar{f},\bar{g}) \in E(f_d^{\pi(t)},g_h^{\pi(t)})\bigr\}.$$
If $g^t \geq \beta$, $(f^t, g^t) \in E$ and then, $F_t$ is non-empty, a contradiction. Thus, $g^t \sm \beta$. Let $(f_d^{\pi(t+1)}, g_h^{\pi(t+1)})$ be the couple's reservation payoffs at the following iteration, and set again $F_{t+1} = E \cap E(f_d^{\pi(t+1)}, g_h^{\pi(t+1)})$. If $F_{t+1}$ is non-empty, they play a uniform equilibrium. Otherwise, since $g_h^{\pi(t+1)} \leq g^t \sm \beta$ (\Cref{teo:the_outside_options_of_a_blocking_pair_are_higher_than_their_payoff}), in particular it holds that $f_d^{\pi(t+1)} \geq \alpha$ and $g_h^{\pi(t+1)}  \sm \beta$. Let $(f^{t+1}, g^{t+1})$ be the new constrained Nash equilibrium payoff found by the oracle. Since pairwise stability implies that $(f^t, g^t) \in E(f_d^{\pi(t+1)}, g_h^{\pi(t+1)})$ (\Cref{teo:the_outside_options_of_a_blocking_pair_are_higher_than_their_payoff}), $g^{t+1} \geq g^t$. In addition, as $F_{t+1}$ is empty, $g^{t+1}  \sm \beta$. Thus, $h$'s sequence of constrained Nash equilibrium payoffs $(g^t)_t$ is non-decreasing and bounded from above by $\beta$. Therefore, the sequence converges to a fixed payoff, and then, the renegotiation process converges as well.
\end{proof}

\begin{remark}
Unlike \Cref{teo:convergence_reneg_process_potential_games,teo:convergence_reneg_process_inf_rep_games}, where we have managed to construct an oracle to ensure the convergence of our renegotiation process, remark \Cref{teo:convergence_reneg_process_zero_sum_games} holds for any choice of oracle.
\end{remark}

\subsection{Pareto-optimality, pairwise stability, and Nash stability.}\label{sec:the_strongest_stability}

Consider a couple who plays a two-player infinitely repeated game. For simplicity, we drop the identity of the agents from the notation. Given a pair of reservation payoffs $(f^0,g^0)$, we recall the oracle designed for the feasibility of infinitely repeated games proof. Let $F := E(f_0,g_0) \cap E$, where $E(f_0,g_0)$ is the set of acceptable payoffs and $E$ is the set of uniform equilibrium payoffs. For any case in which $F$ is non-empty, the oracle picks a $(f_0,g_0)$-feasible uniform equilibrium as constrained Nash equilibrium. Otherwise, the oracle computes a $(f_0,g_0)$-CNE. It is interesting to study in which cases $F$ is empty.

Consider a \textit{matching pennies} game played in mixed strategies (\Cref{tab:matching_pennies_repeated}, also studied in \Cref{ex:matching_pennies}) repeated infinitely many times. The infinitely repeated version of the matching pennies game has $0$ as the only uniform equilibrium payoff. Consider reservation payoffs $(f_0,g_0)$ such that $w \sm f_0 \leq g_0$, as in \Cref{fig:matching_pennies_repeated}. It follows that $E(f_0,g_0)$ is the continuum of values in the blue line. It holds $F = \{0\} \cap E(f_0,g_0) = \varnothing$. The only constrained uniform equilibrium payoff corresponds to the limit point between the red and blue line, that is, when player $1$ gets exactly $f_0$ as payoff.

Any constant-sum game used as stage game for an infinitely repeated game will present the same behavior as matching pennies. However, for most of the rest of the stage games that we may consider, the intersection set $F$ will be non-empty. Let us illustrate this with another example. Let $G$ be a prisoners' dilemma (\Cref{tab:prinsoners_dilemma_repeated}) played in mixed strategies and repeated infinitely many times. 

The punishment levels of the agents being equal to $(0,0)$ (the unique Nash equilibrium payoff), the set of uniform equilibrium payoffs (\Cref{fig:uniform_payoffs}) is,
$$E = \{(\Bar{f},\Bar{g}) \in co(\{(2,2),(-1,3),(3,-1),(0,0)\}) : \Bar{f} \geq 0 \text{ and } \Bar{g} \geq 0\}.$$
Given a pair of reservation payoffs $(f_0,g_0)$ such that there exists at least one feasible contract, \Cref{fig:feasbility_proof_repeated_games} shows the four possible cases studied in the proof of the feasibility of infinitely repeated games.

\begin{minipage}[H]{0.4\linewidth}
\begin{table}[H]
\begin{tabular}{cccc}
& \multicolumn{3}{c}{Player 2}\\\noalign{\vskip 2mm}
\multirow{3}{*}{Player 1 } & & A & B \\
\cmidrule{2-4}
& A & 1 & -1 \\
\cmidrule{2-4}
& B & -1 & 1
\end{tabular}
\caption{Matching pennies}
\label{tab:matching_pennies_repeated}
\end{table}
\end{minipage}\hfill
\begin{minipage}[H]{0.58\linewidth}
\begin{figure}[H]
\centering
\includegraphics[scale = 0.7]{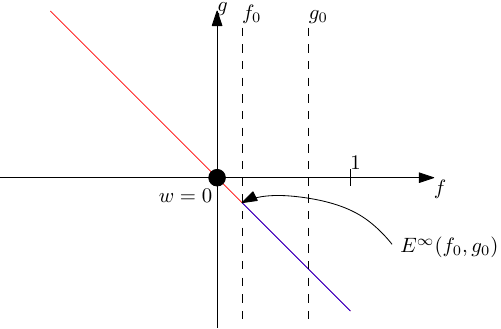}
\caption{Matching pennies}
\label{fig:matching_pennies_repeated}
\end{figure}
\end{minipage}
\medskip

\begin{minipage}[H]{0.5\linewidth}
\begin{table}[H]
\begin{tabular}{cccc}
     & \multicolumn{3}{c}{Player 2}\\\noalign{\vskip 2mm}
     \multirow{3}{*}{\hspace{-0.5cm}Player 1  } & & Cooperate & Betray \\
     \cmidrule{2-4}
     & Cooperate & $2, 2$ & $-1,3$ \\ 
     \cmidrule{2-4}
     & Betray & $3,-1$ & $0,0$
\end{tabular}
\caption{Prisoner's dilemma}
\label{tab:prinsoners_dilemma_repeated}
\end{table}
\end{minipage}\hfill
\begin{minipage}[H]{0.4\linewidth}
\begin{figure}[H]
\centering
\includegraphics[scale = 0.7]{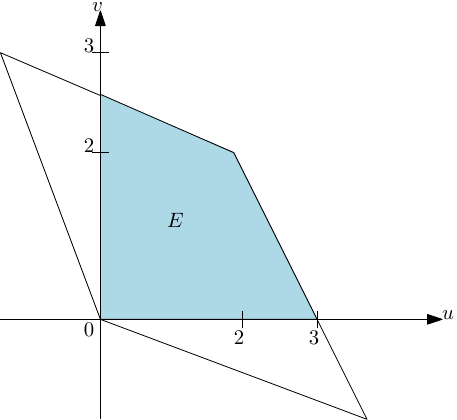}
\caption{Uniform equilibrium payoffs}
\label{fig:uniform_payoffs}
\end{figure}
\end{minipage}
\medskip

\begin{figure}[t]
\centering
\includegraphics[scale = 0.7]{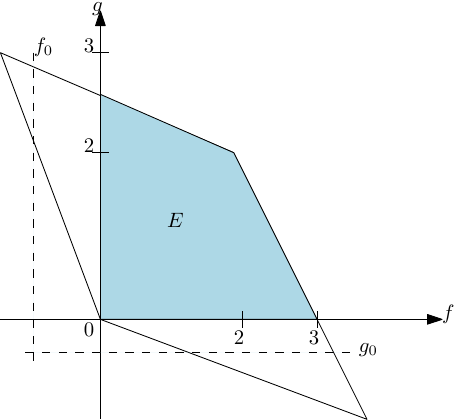}\qquad \includegraphics[scale = 0.7]{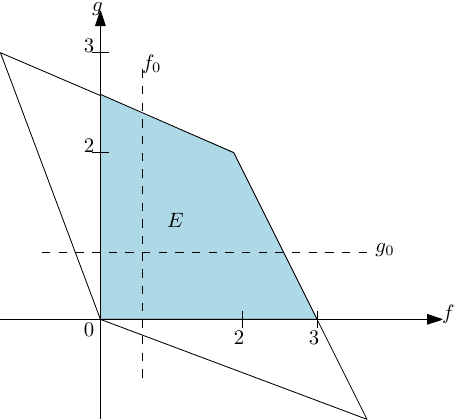}

\includegraphics[scale = 0.7]{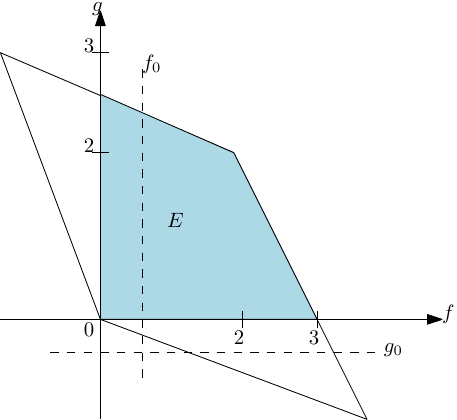}\qquad \includegraphics[scale = 0.7]{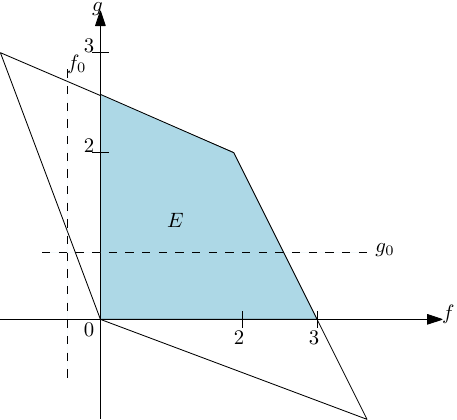}
\caption{Feasibility proof infinitely repeated games}
\label{fig:feasbility_proof_repeated_games}
\end{figure}

Thanks to the non-emptiness of the interior of $E$, given that for $(f_0,g_0)$ there exists at least one feasible strategy profile, the intersection between $E$ and $E(f_0,g_0)$ gives a non-empty set. Therefore, we can always find a $(f_0,g_0)$-feasible uniform equilibrium of the repeated game in which no agent has incentives to deviate. This is not the case for infinitely repeated games with constant-sum games as stage games. To an infinitely repeated game that has a set of uniform equilibrium payoffs with a non-empty interior we call it \textbf{non-degenerated}.

Consider next a matching game $\Gamma = (D_0, H_0, \{G_{d,h}: (d,h) \in D \times H\}, \underline{f}, \underline{g})$. We recall the notion of Nash stability given in \Cref{chapter:one_to_one_matching_games}.
\medskip

\noindent\textbf{\Cref{def:internally_Nash_stable_matching_profile}}. An allocation $\pi = (\mu,\vec{x},\vec{y})$ is \textbf{Nash stable} if for any matched couple $(d,h) \in \mu$, $(x_d,y_h)\in \text{N.E}(G_{d,h})$, i.e, $(x_d,y_h)$ is a Nash equilibrium of $G_{d,h}$.
\medskip

The following theorem proves the equivalence between the models with and without commitment when $\Gamma$ is an infinitely repeated matching game.
\medskip

\begin{theorem}\label{teo:equivalance_commitment_and_no_commitment_repeated_games}
Let $\Gamma$ be a matching game in which each strategic game $G_{d,h}$, for $(d,h) \in D \times H$, is a non-degenerated infinitely repeated game. Let $\pi = (\mu, \vec{\sigma}_D,\vec{\sigma}_H)$ be a pairwise stable allocation, with $\vec{\sigma}_D$ and $\vec{\sigma}_H$ profiles of behavioral strategies. 
Then, there always exists a pairwise stable (\Cref{def:externally_stable_matching_profile}) and Nash stable (\Cref{def:internally_Nash_stable_matching_profile}) allocation $\pi' = (\mu,\vec{\tau}_D,\vec{\tau}_H)$ that weakly Pareto-dominates $\pi$. In addition, an allocation is renegotiation proof allocation if and only if it is Nash stable.
\end{theorem}
\medskip

Remark that $\pi$ and $\pi'$ in the previous theorem have the same matching $\mu$.

\begin{proof}
That renegotiation proof allocations and Nash stable allocations are equivalent comes from the previous discussion. As the games are non-degenerated, for each couple there always exists a feasible uniform equilibrium. Therefore, the allocation is renegotiation proof if and only if the oracle can always choose a uniform equilibrium as constrained Nash equilibrium.

Regarding the first point, consider a couple $(d,h)\in \mu$ and their behavioral strategy profile $(\sigma_d,\sigma_h)$ in $\pi$. The payoff profile $(f_{d,h}(\sigma_d,\sigma_h), g_{d,h}(\sigma_d,\sigma_h)) \in \mathbb{R}^2$ belongs to the convex envelope of the stage game payoffs. Since the game is non-degenerated, there exists a payoff profile $(f_d^*,g_h^*) \geq (f_{d,h}(\sigma_d,\sigma_h), g_{d,h}(\sigma_d,\sigma_h))$ that belongs to the set of uniform equilibrium payoffs. Therefore, there exists a uniform equilibrium $(\tau_d,\tau_h)$ satisfying,
\begin{align*}
    &f_{d,h}(\tau_d,\tau_h) = f_d^* \geq f_{d,h}(\sigma_d,\sigma_h) = f_d(\pi),\\
    &g_{d,h}(\tau_d,\tau_h) = g_h^* \geq g_{d,h}(\sigma_d,\sigma_h) = g_h(\pi).
\end{align*}
We conclude the proof.
\end{proof}

We have remarked that the oracle designed to compute constrained Nash equilibria in infinitely repeated games selects a Pareto-optimal outcome (\Cref{remark:oracle_for_repeated_games_is_always_Pareto}). Consequently, the pairwise stable and renegotiation proof allocations in matching games with non-degenerated infinitely repeated keep the property that all couples play Pareto-optimally within their games. In this setting, we achieve the strongest stability criterion in our model: \textit{Pareto-optimality, pairwise stability, and Nash stability}.

\section{Conclusions}

This article introduces an extension of the stable matching model which endows the agents within the market with strategy sets and payoff functions. The model encompasses a large number of works from the literature. 



Using a \textit{deferred-acceptance with competitions} algorithm, we prove the existence of pairwise stable allocations with and without commitment, under mild conditions. Moreover, when players can commit, a novel \textit{renegotiation proofness} condition allows use to refine drastically the set of pairwise stable solutions (sometimes from the continuum to a single prediction).


Renegotiation proof stable allocations are characterized as those in which all matched pairs play (or sign bilateral contracts that satisfy) a participation-\textit{constrained Nash equilibrium} (CNE) condition. After proving that some games do not admit a CNE, we identify a large class admitting a CNE, including strictly competitive, potential and infinitely repeated games. We then designed a renegotiation process that converges to a renegotiation proof pairwise stable allocations in all those classes of games. Moreover, it can be proved that, over the same classes, our two algorithms -DAC and the renegotiation process- converge in polynomial time whenever the underlying strategic games are bi-linear \citep{garridolucero:tel-04167090}.


The characterization of all feasible  strategic games (i.e. those that admit a CNE for all outside options) is an open problem that may be of interest beyond our framework, e.g. in contract theory, but also in a two stage matching game model with commitment as studied by \cite
{noldeke2015investment}. Actually, suppose that at stage 1, each player must choose an action (an investment, a type or characteristic) that will advantage him/her in the matching market in the future (e.g. a diploma, programming skills, reputation, patents). At stage 2, players are matched --or remains single-- by some stable matching algorithm such as the Gale-Shapley's deferred-acceptance where the type $d$ players propose. The players payoffs are given as follows: if agent $d$ has chosen at stage 1 $x_d$ and is matched with agent $h$ at stage 2 who has chosen action $y_h$ at stage 1, the payoff of $d$ is $f_d(x_d,y_h)$ and of $h$ is $g_h(x_d,y_h)$. One may wonder if this game admits a Nash equilibrium. \Cref{ex:non_feasible_game} shows that it may not exist and that the feasibility condition is necessary for the Nash equilibrium existence. The question of whether feasibility is sufficient remains open. 

The \cite{noldeke2015investment} model is different from ours as players' actions are chosen and fixed once for all at stage 1, and will not be changed during the matching process or be adapted to the partner. In our model, however, a player may adapt his behavior to the partner. One could combine both!

The generalization to one-to-many markets, where hospitals can accept multiple doctors, is a natural next step for our model. In such cases, conditions like substitutability are necessary to ensure the existence of a pairwise/core stable allocation \citep{garridolucero:tel-04167090}. Similarly, extending the model to dynamic environments offers an interesting research direction with numerous potential applications, such as the job market, where a matching game is played repeatedly over time, allowing agents to change or experiment with new partners.

In line with the overall spirit of this paper, our infinitely repeated game model (\Cref{sec:infinitely_repeated_games_are_feasible}) assumes that agents commit not only to certain strategies within the repeated game but also to remaining paired with the same partner. However, it has been shown that commitment to strategies can be relaxed if the associated one-shot game is generic (\Cref{sec:the_strongest_stability}). By further relaxing the assumption of partner commitment, our results can be interpreted as demonstrating the existence of \textit{dynamically pairwise stable stationary equilibria} in general dynamic matching models where players are \textit{impatient}, stage games are generic, and agents are free to change partners during play. This can be achieved because, at any point, remaining with the current partner is pairwise stable — our equilibrium construction ensures that players receive at least the payoff they could obtain with another partner. Additionally, since the continuation payoff at any moment along the equilibrium path matches the initial equilibrium payoff, the pairwise stability condition holds throughout the repeated games.

This leads to key questions in the spirit of the Folk theorem: (1) how can we characterize the set of dynamically pairwise stable equilibrium payoffs, particularly when players can change partners, and (2) under what conditions does this set correspond to the limiting stable discounted equilibrium payoffs as the discount factor approaches 1?

\begin{appendices}
\section{Strictly competitive games}\label{sec:feasibility_strictly_competitive_games} 

First of all, we introduce strictly competitive games.
\medskip

\begin{definition}\label{def:strictly_competitive_game}
\textup{A two-player game $G = (X,Y,f,g)$ is \textbf{strictly competitive} if for any strategy profile $(x,y) \in X\times Y$, it holds,
\begin{align*}
&\forall x' \in X : f(x',y) \bi f(x,y), g(x',y) \sm g(x,y), \text{ and},\\
&\forall y' \in Y : g(x,y') \bi g(x,y), f(x,y') \sm f(x,y).
\end{align*}}
\end{definition}

We denote $\mathcal{S}$ the class of strictly competitive games à la \cite{aumann1961almost}, that is, all strictly competitive games obtained as monotone transformations of zero-sum games. To prove \Cref{teo:Class_of_feasible_games} we use Aumann's conjecture that $\mathcal{S}$ covers the entire class of competitive games, although the proof is known only for the finite case \citep{adler2009note}. Consider a strictly competitive game $G = (X,Y,f,g)$, with $X,Y$ compact sets and $f,g$ continuous payoff functions. Let $\varphi, \phi$ be increasing functions such that the game $G' = (X,Y,u)$ is a zero-sum game where $f = \varphi \circ u$ and $g = \phi \circ u$. Nash equilibria of $G$ and $G'$ coincide, and Nash equilibrium payoffs are the image through the increasing functions from one game to another. In particular, if $w$ is the value of $G'$, then $(\varphi^{-1}(w), \phi^{-1}(w))$ is a Nash equilibrium payoff of $G$. 

Let $(f_0,g_0)$ be outside options of the players in $G$, and let $(x^*,y^*)$ be a $(f_0,g_0)$-constrained Nash equilibrium. Consider the corresponding outside options in $G'$ given by $f_0' := \varphi(f_0)$, $g_0' := -\phi(g_0)$. Indeed, $f_0'$ and $g_0'$ are outside options for the players in their zero-sum game as, for any $(x,y) \in X \times Y$ such that $f_0 \leq f(x,y)$ and $g_0 \leq g(x,y)$, it holds 
$$f'_0 \leq \varphi(f_0) \leq \varphi(f(x,y)) = - \phi(g(x,y)) \leq g_0'.$$
Also, $(x^*,y^*)$ is $(f_0',g_0')$-feasible in $G'$, and it is direct that $(x^*,y^*)$ is a $(f_0',g_0')$-constrained Nash equilibrium, as increasing functions preserve inequalities. From the proof of zero-sum games' feasibility, we conclude the following result.
\medskip

\begin{theorem}\label{teo:strictly_competitive_games_are_feasible}
Let $G = (X,Y,f,g)$ be a strictly competitive game, $G' = (X,Y,u)$ a zero-sum game, and $\varphi, \phi$ increasing functions such that the game $f = \varphi \circ u$, $g = \phi \circ u$. Suppose that $G'$ has a value $w$ (if and only if $G$ has a Nash equilibrium). Then, given $f_0, g_0$ outside options in $G$, which admit a feasible strategy profile, there always exists a $(f_0,g_0)$-CNE $(x^*,y^*)$ of $G$. In addition, it holds,
\begin{align*}
f(x^*,y^*) &= \text{median}\{f_0,\varphi^{-1}(-\phi(g_0)), \varphi^{-1}(w)\}\\
g(x^*,y^*) &= \text{median}\{\phi^{-1}(-\varphi(f_0)),g_0, \phi^{-1}(-w)\}
\end{align*}
\end{theorem}

\begin{remark}
\textup{As for zero-sum games, the selection of constrained Nash equilibria is always \textbf{Pareto-optimal} for strictly competitive games.
}\end{remark}
\medskip

Regarding the convergence of the renegotiation process, this is a corollary of the one for zero-sum games. 
\medskip

\begin{theorem}[\textbf{Convergence renegotiation process strictly competitive ga- mes}]\label{teo:convergence_reneg_process_str_comp_games}
For any oracle, the renegotiation process converges for strictly competitive games in $\mathcal{S}$ with an equilibrium.
\end{theorem}
\medskip

Shapley-Shubik's and Demange-Gale's models can be mapped into a matching game in which all strategic games $G_{d,h}$ are included in the class $\mathcal{S}$ of strictly competitive games (see Section \ref{sec:shapley_shubik_gale_demange_models}). Our results, therefore, apply directly to their works proving the existence of allocations that are not only pairwise stable but also renegotiation proof. The refinement induced by renegotiation proofness crucially depends on the choice of the strategic games $G_{d,h}$. For example, if we model the game between a buyer and a seller as an ultimatum game \citep{abreu2000bargaining} where the buyer is the first proposer, she gets all the surplus, while when the first proposer is the seller, she is the one who gets all the surplus. However, if the game is an alternative offer bargaining game \citep{osborne2019bargaining}, the surplus is shared equally.

\end{appendices}

\section*{Declarations}

\textbf{Funding and Competing interests}. All authors certify that they have no affiliations with or involvement in any organization or entity with any financial interest or non-financial interest in the subject matter or materials discussed in this manuscript.

\bibliography{sn-bibliography}

\end{document}